\documentclass[]{interact}

\usepackage{epstopdf}
\usepackage[caption=false]{subfig}
\usepackage{amsthm}
\usepackage{graphicx}
\usepackage{amssymb}
\usepackage{colortbl}
\usepackage{arydshln}
\usepackage{amsmath}
\usepackage{float}
\usepackage{widetext}
\usepackage[vlined,boxed,ruled,linesnumbered]{algorithm2e}
\usepackage{algpseudocode}
\usepackage{multirow}
\usepackage{extarrows}
\usepackage{mathrsfs}
\usepackage{upgreek}

\usepackage[natbibapa,nodoi]{apacite}
\theoremstyle{plain}
\newtheorem{Theorem}{Theorem}
\newtheorem{Lemma}{Lemma}
\newtheorem{Assumption}{Assumption}
\newtheorem{Problem}{Problem}
\newtheorem{Proposition}{Proposition}
\theoremstyle{definition}
\newtheorem{Definition}{Definition}
\theoremstyle{remark}
\newtheorem{Remark}{Remark}

\newcommand{\wesley}[1]{{\textcolor{black}{#1}}}

\begin{document}

\title{Observer-Based Safety Monitoring of Nonlinear Dynamical
Systems with Neural Networks via Quadratic
Constraint Approach}

\author{Tao Wang, Yapeng Li, Zihao Mo, Wesley Cooke and Weiming Xiang 
\thanks{This work was supported by the National Science Foundation, under NSF CAREER Award 2143351, NSF CNS Award no. 2223035, and NSF IIS Award no. 2331938. }    
\thanks{Tao Wang and Yapeng Li are with the School of Electrical Engineering, Southwest Jiaotong University, Chengdu, China. Emails: wangtao618@126.com (T. Wang), liyapeng@my.swjtu.edu.cn (Y. Li)}
\thanks{Zihao Mo, Wesley Cooke, and Weiming Xiang are with the School of Computer and Cyber Sciences, Augusta University, Augusta GA 30912, USA. Emails: zmo@augusta.edu (Z. Mo), wxiang@augusta.edu (W. Xiang), wcooke@augusta.edu (W. Cooke)}
}

\maketitle 

\begin{abstract}
The safety monitoring for nonlinear dynamical systems with embedded neural network components is addressed in this paper. The interval-observer-based safety monitor is developed consisting of two auxiliary neural networks derived from the neural network components of the dynamical system. Due to the presence of nonlinear activation functions in neural networks, we use quadratic constraints on the global sector to abstract the nonlinear activation functions in neural networks. By combining a quadratic constraint approach for the activation function with Lyapunov theory, the interval observer design problem is transformed into a series of quadratic and linear programming feasibility problems to make the interval observer operate with the ability to correctly estimate the system state with estimation errors within acceptable limits. The applicability of the proposed method is verified by simulation of the lateral vehicle control system.
\end{abstract}


\section{Introduction}

Complex dynamical systems, such as autonomous vehicles and various cyber-physical systems (CPS), have been greatly benefiting from the fast advancement of artificial intelligence (AI) and machine learning (ML)  technologies. Many new theories have been proposed on this basis, such as stable neural network controllers and observers \citep{6748099,287122,7440845}, adaptive neural network controllers \citep{8031467,8681071} and various neural network controllers \citep{HUNT19921083}. Real-time monitoring of these dynamical systems embedded with neural network components is essential to ensure the system's safety. External inputs may have adversarial effects on the normal working state of the system; even with the most advanced neural networks, imperceptible perturbations in the input may lead to an erroneous result \citep{8099500}. In addition, these systems are highly susceptible to erroneous outputs if they are subjected to adversarial attacks, which can have serious safety consequences. Therefore, to ensure the security of dynamical systems embedded in neural networks, it is essential to develop a technique that can monitor the operational state of dynamical systems in real time.

Most current approaches to safety or security verification take the form of offline computation. In general, verification using offline calculation requires a large amount of computational resources due to its high computational complexity. For example, for a type of neural networks with the activation function of rectified linear unit (ReLU), the safety verification problem can be represented as various complex computational problems. Based on polyhedral operations, a geometric computation method is proposed to obtain the exact output set of the neural network using ReLU activation function \citep{8431048,xiang2017reachable}. Based on those results, the methods in \citep{tran2019star,8807491} extended it by proposing a novel approach with the aid of a specific convex set representation called star sets, which greatly improved scalability. 
A mixed-integer linear programming (MILP) method to validate neural networks \wesley{was proposed} in \citep{lomuscio2017approach}.
The work \citep{dutta2019reachability} focuses on neural networks with ReLU activation functions; \wesley{they used} a Taylor-model-based flowpipe construction scheme \wesley{and replaced} the neural network feedback \wesley{with} a polynomial mapping approach for a small fraction of the input to obtain an over-approximated reachable set. In addition, this method can be extended to other activation units after processing by segmental linearization \citep{dutta2018output}. The work \citep{8318388} introduces a simulation-based approach to output reachability estimation for neural networks with common activation functions. This paper \citep{9093970} takes the dynamic system embedded in the feedforward neural network named multilayer perceptrons (MLPs) as the research object, and develops a recursive algorithm with over-approximating the reachable set of the closed-loop system. The security verification of the system is achieved by checking the emptiness of the intersection between the insecure sets and the over-approximation of the reachable sets.

It is worth noting that the open-loop computational structure of these offline methods makes them quite challenging to implement in online settings. On the other hand, offline methods are difficult to detect system security issues in a timely manner, and the system state and parameters may differ from run-time when offline. Therefore, developing an online security monitoring method is very important. For this reason, inspired by observer design theory, we propose an alternative solution to design closed-loop systems for run-time monitoring based on instantaneous measurements of the system. We resort to develop interval observers for dynamical systems with neural networks. The interval observer can estimate the upper and lower bounds of the operating state trajectory of the dynamical system in real-time, which can achieve real-time safety monitoring of the dynamical system \citep{efimov2016design,chebotarev2015interval,bolajraf2011robust,6908993,9279463}. As shown in \citep{9380552}, unlike the general interval observer design approach, the observer gains as well as auxiliary neural networks  have to be designed through a series of optimization problems to ensure that the interval observer can correctly estimate the upper and lower state bounds and a suitable estimation error. The design of the auxiliary neural networks in the interval observer is also necessary to simulate the behavior of the neural network in the original system for better state estimation. The work \citep{9147468} applies interval observers to the safety monitoring of the state of charge (SOC) of lithium-ion batteries. The coupled equivalent circuit-thermal model is adopted in this paper, avoiding the complex structure and calculation caused by the traditional model with electrically and thermally coupled parallel connection of cells. The innovation of the work lies in considering cell heterogeneity as the uncertainty bounding functions and achieving the separation of the state number of interval observers from the number of parallel batteries. 

During the design of interval observers, it is challenging to apply classical control theory, such as Lyapunov theory, for analyzing dynamical systems embedded in neural network components due to the various types of nonlinear activation functions in neural networks. A popular approach is using quadratic constraints (QCs) to abstract the nonlinear activation functions in neural networks. The work \citep{4267699} analyzes the stability of the feedback loop, including neural networks, by replacing the nonlinear and time-varying components of the neural networks with integral quadratic constraints (IQCs). Quadratic constraints are used to abstract the nonlinear activation functions and projection operators in neural network controllers in \citep{9304296}, enabling the reachability analysis of closed-loop systems with neural network controllers. The approach in \citep{9301422} uses quadratic constraints to abstract various properties of the activation function, such as bounded slope, monotonicity, and cross-layer repetition, thus formulating the safety verification problem for neural networks as the SDP feasibility problem. In addition, the characterization of the input-output of neural networks through quadratic constraints allows other issues to be solved, such as the input-output sensitivity analysis of neural networks \citep{8318388}, safety verification and robustness analysis \citep{9301422}, Lipschitz constant estimation of feedforward neural networks \citep{fazlyab2019efficient}, etc.

Synthesizing the previous discussions, the main contributions of this paper are as follows: (1) A global quadratic constraint formulation method for error dynamic systems is discussed; (2) A novel interval observer design method is proposed for the nonlinear dynamical systems with neural networks, and its core contribution is to abstract the nonlinear activation function of neural networks by the quadratic constraints method, so that some control theories applicable to linear systems can also be applied to the nonlinear dynamical systems with neural networks in this paper.

The rest of the paper is organized as follows. In Section II, the system and problem formulation under discussion are presented. The main findings are given in Section III, where the design methods for quadratic constraints on the activation function and auxiliary neural networks are presented, and the design of the interval observer gains $\underline{L}$ \wesley{and} $\overline{L}$ is represented in the form of a series of convex optimization problems. The conclusion obtained is applied to a lateral control system for vehicles in Section IV. In Section V, conclusions and future research directions are given.

\emph{Notations:} In this paper, the notation $\mathbb{R}$ represents real numbers, and $\mathbb{R}_+$ is defined by $\mathbb{R}_+=\{\tau\in\mathbb{R},\tau\geq0\}$. The notation $\mathbb{R}^n$ represents the vector space of all $n$-tuples of real numbers, and $\mathbb{R}^{n\times n}$ is
the space of $n \times n$ matrices with real entries. The superscript ``$T$" denotes the matrix transpose.
The block diagonal matrix is denoted by the symbol $diag\{\cdots\}$. The notation $I_n\in\mathbb{R}^{n\times n}$ denotes the $n$-dimensional identity matrix.  Given a matrix $A\in\mathbb{R}^{m\times n}$, $\Vert A\Vert$ denote its Frobenius norm. 
For two vectors $x_1,x_2\in\mathbb{R}^n$ or matrices $A_1,A_2\in\mathbb{R}^{n\times n}$, 
the relations $x_1<x_2$ and $A_1<A_2$ are interpreted elementwisly. 
The relation $Q\succ0$ ($Q\prec0$) means that $Q\in\mathbb{R}^{n\times n}$ is positive (negative) definite. In addition, $Q>0$ ($Q\geq 0$) means that all elements in this matrix $Q\in\mathbb{R}^{n\times n}$ are positive (nonnegative). $\mathbb{M}_n\in\mathbb{R}^{n\times n}$ is defined as the collection of all $n$-dimensional Metzler matrices.

\section{System Description and Problem Formulation}
\subsection{System Description}

In this paper, we consider a class of learning-enabled nonlinear dynamical systems embedded with neural networks in the following form
\begin{equation}
\begin{aligned}
\left\{
\begin{array}{ll}
\dot{x}(t)=f(x(t),u(t),\Phi(x(t))) \\
y(t)=g(x(t))
\end{array}
\right. ,
\end{aligned}
\label{Formula 1}
\end{equation}
where $x\in \mathbb{R}^{n_x}$, $u(t)\in\mathbb{R}^{n_u}$ and $y\in \mathbb{R}^{n_y}$ are the state vector, input and output of the system, respectively. $f:\mathbb{R}^{n_x+n_u} \to \mathbb{R}^{n_x}$ and $g: \mathbb{R}^{n_x} \to \mathbb{R}^{n_y}$ are nonlinear functions. $\Phi:\mathbb{R}^{n_x} \to \mathbb{R}^{n_x}$ is the neural network component. Without causing ambiguity, we omit the time index $t$ in some of the variables.

Specifically, this work considers a class of dynamical systems embedded with neural networks, which have the form of a Lipschitz nonlinear model as 
\begin{equation}
\begin{aligned}
\mathcal{L}:\left\{
\begin{array}{ll}
\dot{x}=Ax+B_\Phi \Phi(x)+B_u u(t)+f(x) \\
y=Cx
\end{array}
\right. ,
\end{aligned}
\label{Formula 2}
\end{equation}
where $A\in\mathbb{R}^{n_x\times n_x}$, $B_\Phi\in\mathbb{R}^{n_x\times n_{L+1}}_+$, $B_u\in\mathbb{R}^{n_x\times n_{u}}_+$, $C\in\mathbb{R}^{n_y\times n_x}$ and $f(x)$ is a Lipschitz nonlinear function satisfying the following Lipschitz inequality
\begin{equation}
\begin{aligned}
\Vert f(x_1)-f(x_2)\Vert\leq\beta\Vert x_1-x_2\Vert,~\beta>0.
\end{aligned}
\label{Formula 3}
\end{equation}


\begin{Remark}\label{Remark 2}
Many nonlinear systems in the form of $\dot{x}=f(x,u,\Phi(x))$ can be represented in the form of (\ref{Formula 2}) if $f$ is differentiable with respect to $x$ and $u$. The neural network $\Phi(x)$ is the interval component that affects the behavior of the system. For instance, the model (\ref{Formula 2}) represents a state feedback closed-loop system if the neural network $\Phi(x)$ is trained as a feedback controller.
\end{Remark}

For the system (\ref{Formula 2}), there are two sources of uncertainty: the initial values for state $x(0)$ \wesley{and} the instantaneous values of input $u(t)$. We assume that all these uncertainties belong to the known interval as shown in the following assumption.
\begin{Assumption}\label{Assumption 4}
Let $\underline{x}(0)\leq x(0)\leq \overline{x}(0)$ for some known $\underline{x}(0)$ \wesley{and} $\overline{x}(0)\in\mathbb{R}^{n_x}$, and let the known bounded functions $\underline{u}$ \wesley{and} $\overline{u}$ such that
$\underline{u}(t)< u(t)<\overline{u}(t), \forall t\geq0$.
\end{Assumption}

Suppose that the nonlinear function $f(x)$ has the following properties.
\begin{Assumption}\label{Assumption 1}
Suppose there exist functions $\underline{f},\overline{f}:\mathbb{R}^{2n_x}\to\mathbb{R}^{n_x}$ such that
\begin{equation}
\underline{f}(\underline{x},\overline{x})\leq f(x)\leq\overline{f}(\underline{x},\overline{x}) ,
\label{Formula 4}
\end{equation}
holds for any $\underline{x}\leq x\leq\overline{x}$.
\end{Assumption}

\begin{Remark}\label{Remark 6}
Assumptions \ref{Assumption 4} and \ref{Assumption 1} emphasize that the initial state, the input signal and the nonlinear function of the original system, must numerically lie in the interval consisting of the initial state, the input signal and the nonlinear function of the interval observer, respectively. This is to ensure that the interval observer can correctly achieve the interval estimate for the state of the original system, which means $\underline{x}\leq x\leq\overline{x}$, in other words, to ensure that the error system is a positive system.
\end{Remark}

\begin{Assumption}\label{Assumption 2}
Suppose there exist scalars $\underline{a}_1,\overline{a}_1,\underline{a}_2,\overline{a}_2\in\mathbb{R}_+$ and vectors $\underline{\rho},\overline{\rho}\in\mathbb{R}_+^{n_x}$ such that
\begin{equation}
\begin{aligned}
f(x)-\underline{f}(\underline{x},\overline{x})\leq \underline{a}_1(x-\underline{x})+\underline{a}_2(\overline{x}-x)+\underline{\rho} ,\\
\overline{f}(\underline{x},\overline{x})-f(x)\leq \overline{a}_1(x-\underline{x})+\overline{a}_2(\overline{x}-x)+\overline{\rho} ,
\end{aligned}
\notag
\end{equation}
holds for the nonlinear functions $\underline{f}(\underline{x},\overline{x}),f(x),\overline{f}(\underline{x},\overline{x})$ defined in Assumption \ref{Assumption 1}.
\end{Assumption}

\begin{Remark}\label{Remark 1}
Under the Lipschitz condition (\ref{Formula 3}), the estimation of parameters $\underline{a}_1$, $\overline{a}_1$, $\underline{a}_2$, $\overline{a}_2$, $\underline{\rho}$, $\overline{\rho}$ in Assumption \ref{Assumption 2} can be obtained through routine calculation, and the detailed estimation procedures can be found in Lemma 6 of \citep{ZHENG2016167}.
\end{Remark}

An $L$-layer feedforward neural network $\Phi(x):\mathbb{R}^{n_0}\to \mathbb{R}^{n_{L+1}}$ are considered in this work, which is defined by the following recursive equation
\begin{equation}
\mathcal{N}: \left\{
\begin{array}{ll}
\omega^{[0]}=x(t) &
\\
v^{[l]}=W^{[l]}\omega^{[l-1]}+b^{[l]} & l=1,\dots,L
\\
\omega^{[l]}=\phi^{[l]}(v^{[l]})&l=1,\dots,L
\\
\Phi(x)=W^{[L+1]}\omega^{[L]}+b^{[L+1]}
\end{array}
\right. ,
\label{Formula 5}
\end{equation}
where $\omega^{[l]}\in\mathbb{R}^{n_l}$ denotes the output from the $l^{th}$ layer with $n_l$ neurons of the neural network. $v^{[l]}\in\mathbb{R}^{n_l}$ denotes the input to the activation function of the $l^{th}$ layer of the neural network.  $\Phi(x)\in\mathbb{R}^{n_{L+1}}$ \wesley{is} the output of the neural network feedback controller. $W^{[l]}\in\mathbb{R}^{n_l\times n_{l-1}}$ and $b^{[l]}\in\mathbb{R}^{n_l}$ represent the weight matrix and bias vector of the $l^{th}$ layer neural network, respectively. In the $l^{th}$ layer neural network, for vectors $v^{[l]}=[v^{[l]}_1,v^{[l]}_2,\dots,v^{[l]}_{n_l}]^T$, we define $\phi^{[l]}=[\psi^{[l]},\psi^{[l]},\dots,\psi^{[l]}]^T$ to be the series of activation functions and a single activation function is $\psi$, where $\phi^{[l]}(v^{[l]})$ is the action on each element in the vector, i.e.
\begin{equation}
\begin{aligned}
\phi^{[l]}(v^{[l]})=[\psi^{[l]}(v^{[l]}_1),\psi^{[l]}(v^{[l]}_2),\dots,\psi^{[l]}(v^{[l]}_{n_l})]^T
\end{aligned} .
\notag
\end{equation}

Here, the following assumptions about the activation function are given.
\begin{Assumption}\label{Assumption 3}
Suppose that for activation functions $\psi^{[l]},~l=1,\dots,L$, the following properties hold:
\begin{itemize}
    \item Any two scalars $x_1$ and $x_2$ are given, there must be a scalar $\alpha>0$ such that
\begin{equation}
\begin{aligned}
\vert\psi^{[l]}(x_1)-\psi^{[l]}(x_2)\vert\leq\alpha\vert x_1-x_2\vert,~ \forall l=1,\dots,L .
\end{aligned}
\label{Formula 6}
\end{equation}
\item Any two scalars $x_1\leq x_2$ are given, and we have
\begin{equation}
\begin{aligned}
\psi^{[l]}(x_1)\leq\psi^{[l]}(x_2),~\forall l=1,\dots,L .
\end{aligned}
\label{Formula 7}
\end{equation}
\end{itemize}
\end{Assumption}

\begin{Remark}\label{Remark 3}
Assumption \ref{Assumption 3} above applies to the most common activation functions, such as ReLU, sigmoid, tanh, and leaky ReLU.
For condition (\ref{Formula 6}), the $\alpha$ can be obtained by the maximum Lipschitz constant of all $\psi^{[l]}$. 
The condition (\ref{Formula 7}) is satisfied because the common activation functions are monotonically increasing. Without loss of generality, we suppose that the activation functions are the same in each layer.
\end{Remark}

\subsection{Problem Formulation}
Our proposed solution to the problem of safety monitoring of  neural-network-embedded systems is to design a state estimator which is capable of estimating the upper and lower bounds of the state variable $x(t)$ to monitor the operation status of the system in real time. Information about the system $\mathcal{L}$ in the form of (\ref{Formula 2}) being used for the estimator design includes: the system matrices $A$, $B_\Phi$, $B_u$, $C$, the nonlinear function $f$, the neural network $\Phi$, namely the weight matrix $\{W_l\}^{L+1}_{l=1}$, the bias vector $\{b_l\}^{L+1}_{l=1}$, the known bounded functions $\underline{u},\overline{u}$ and the output $y(t)$. The run-time safety estimator design problem can be 
expressed as follows.

\begin{Problem}\label{Problem 1}
For a dynamical system embedded with neural networks in the form of (\ref{Formula 2}), how can we design a run-time safety state estimator such that its \wesley{instantaneous} state estimates\wesley{,} $\underline{x}$ and $\overline{x}$\wesley{,} satisfy $\underline{x}\leq x\leq\overline{x},\forall t\geq0$?
\end{Problem}

To solve the above problem, we consider the development of a run-time safety state estimator in the form of the Luenberger interval observer
\begin{equation}
\left\{
\begin{array}{ll}
\underline{\dot{x}}=(A-\underline{L}C)\underline{x}+\underline{L}y+B_\Phi\underline{\Phi}(\underline{x},\overline{x})+B_u\underline{u}(t)+\underline{f}(\underline{x},\overline{x})\\
\dot{\overline{x}}=(A-\overline{L}C)\overline{x}+\overline{L}y+B_\Phi\overline{\Phi}(\underline{x},\overline{x})+B_u\overline{u}(t)+\overline{f}(\underline{x},\overline{x})
\end{array}
\right. 
\label{Formula 8}
\end{equation}
where the initial state of the interval observer satisfies $\underline{x}(0)\leq x(0)\leq\overline{x}(0)$, $u(t)$ satisfies $\underline{u}(t)< u(t)<\overline{u}(t), \forall t\geq0$, as shown in Assumptions \ref{Assumption 4}, and $\underline{f}(\underline{x},\overline{x})$, $\overline{f}(\underline{x},\overline{x})$ satisfy Assumptions \ref{Assumption 1} and \ref{Assumption 2}. The auxiliary neural networks $\underline{\Phi}(\underline{x},\overline{x})$ \wesley{and} $\overline{\Phi}(\underline{x},\overline{x})$ and the observer gains $\underline{L}$ \wesley{and} $\overline{L}$ are to be determined.

Here, let the error state $\underline{e}=x-\underline{x},\overline{e}=\overline{x}-x$, so that we can obtain the expression for the error dynamical system in the following form
\begin{equation}
\left\{
\begin{array}{ll}\underline{\dot{e}}=(A-\underline{L}C)\underline{e}+B_\Phi\Delta\underline{\Phi}+B_u(u-\underline{u})+f(x)-\underline{f}(\underline{x},\overline{x})\\
\dot{\overline{e}}=(A-\overline{L}C)\overline{e}+B_\Phi\Delta\overline{\Phi}+B_u(\overline{u}-u)+\overline{f}(\underline{x},\overline{x})-f(x)
\end{array}
\right.
\label{Formula 9}
\end{equation}
where $\Delta\underline{\Phi}=\Phi(x)-\underline{\Phi}(\underline{x},\overline{x}),\Delta\overline{\Phi}=\overline{\Phi}(\underline{x},\overline{x})-\Phi(x)$, the initial state of the error system satisfy $\underline{e}(0)\geq0$ \wesley{and} $\overline{e}(0)\geq0$.

We find that the instantaneous estimates of the interval observer satisfy $\underline{x}(t)\leq x(t)\leq\overline{x}(t),\forall t\geq0$ if we can make the state variable $\underline{e}(t)\geq0,~\overline{e}(t)\geq0,~\forall t\geq0$.
Thus, Problem \ref{Problem 1} can be further formulated as follows.

\begin{Problem}\label{Problem 2}
For a dynamical system embedded with neural networks in the form of (\ref{Formula 2}), how can we design the observer gains $\underline{L}$ \wesley{and}  $\overline{L}$\wesley{,} and the auxiliary neural networks $\underline{\Phi}(\underline{x},\overline{x})$ \wesley{and} $\overline{\Phi}(\underline{x},\overline{x})$ in the interval observer (\ref{Formula 8}) such that error state instantaneous estimates $\underline{e}(t)\geq0$ and $\overline{e}(t)\geq0,~\forall t\geq0$ in error dynamical system (\ref{Formula 9})?
\end{Problem}

To solve Problem \ref{Problem 2}, we review the conclusions related to positive systems.

\begin{Definition}\label{Definition 1}
If all elements outside the main diagonal of a matrix $A\in\mathbb{R}^{n\times n}$ are nonnegative, then $A\in\mathbb{M}_n$.
\end{Definition}

\begin{Lemma}\citep{WANG2022}\label{Lemma 3}
The matrix $PA\in\mathbb{M}_n$ still holds if $P$ is a diagonal positive definite matrix and $A\in\mathbb{M}_n$.
\end{Lemma}

\begin{Lemma}\citep{efimov2016design}\label{Lemma 1}
Considering a system in the form of $\dot{x}(t)=Ax(t)+d(t)$, for $A\in\mathbb{M}_{n}$, the state $x(t)$ is elementwise nonnegative for all $t\geq0$ if $x(0)\geq0$ and $d(t)\in\mathbb{R}_+^n$, and the system is called cooperative.
\end{Lemma}

According to Lemma \ref{Lemma 1}, we propose the following proposition as the solution to Problem \ref{Problem 2}, provided that $x(t)$ and $u(t)$ satisfy Assumptions \ref{Assumption 4} and $\underline{f}(\underline{x},\overline{x})$ \wesley{and} $\overline{f}(\underline{x},\overline{x})$ satisfy Assumptions \ref{Assumption 1} and \ref{Assumption 2}.

\begin{Proposition}\label{Proposition 1}
Problem \ref{Problem 2} can be solved if the observer gains, $\underline{L}$ \wesley{and} $\overline{L}$, and the auxiliary neural networks, $\underline{\Phi}(\underline{x},\overline{x})$ \wesley{and} $\overline{\Phi}(\underline{x},\overline{x})$, satisfy the following conditions
\begin{align} \label{proposition_1}
A-\underline{L}C\in\mathbb{M}_{n_x} ,\\
\label{proposition_4}
A-\overline{L}C\in\mathbb{M}_{n_x}  ,\\
\Phi(x)-\underline{\Phi}(\underline{x},\overline{x})\in\mathbb{R}_+^{n_{L+1}}\label{proposition_2} ,\\
\overline{\Phi}(\underline{x},\overline{x})-\Phi(x)\in\mathbb{R}_+^{n_{L+1}} .\label{proposition_3}
\end{align}
\end{Proposition}

\begin{proof}
According to Assumption \ref{Assumption 4}, \ref{Assumption 1}, it is clear that $f(x)-\underline{f}(\underline{x},\overline{x})\in\mathbb{R}_+^{n_x}$, $x(0)-\underline{x}(0)\in\mathbb{R}_+^{n_x}$ and $B_u(u-\underline{u})\in\mathbb{R}_+^{n_x}$. Since $B_\Phi(\Phi(x)-\underline{\Phi}(\underline{x},\overline{x}))\in\mathbb{R}_+^{n_x}$ holds and $A-\underline{L}C\in\mathbb{M}_{n_x}$, according to Lemma \ref{Lemma 1}\wesley{,} we can conclude $\underline{e}(t)\geq0,~\forall t\geq0$. The same can be said for $\overline{e}(t)\geq0,\forall t\geq0$. Thus the proof is complete.
\end{proof}

It is worth noting that the conditions in Proposition \ref{Proposition 1} hold only to prove that $\underline{e}(t)\geq0, ~\overline{e}(t)\geq0,\forall t\geq0$. Under the conditions that Proposition \ref{Proposition 1} holds, it is possible that $\lim_{t \to \infty}\underline{e}(t)=\infty$ and $\lim_{t \to \infty}\overline{e}(t)=\infty$ happen. Although the interval observer (\ref{Formula 8}) can provide estimated boundaries of the states of the system (\ref{Formula 2}), the estimation error can be extremely large making the estimates meaningless. Therefore, the concept of practical stability, which is related to the boundedness of the system states as time grows, is introduced.

\begin{Lemma}\citep{ge2004adaptive}\label{Lemma 2}
Considering the system (\ref{Formula 2}), if there exists a continuous Lyapunov function $V(x)$ satisfying $a_1(\parallel x \parallel)\leq V(x)\leq a_2(\parallel x \parallel)$, making $\dot{V}(x)\leq -c_1V(x)+c_2$, where $a_1$ \wesley{and} $a_2$ are class $\mathcal{K}$ functions of the state $x$, and $c_1$ \wesley{and} $c_2$ are positive constants, then the solution $x(t)$ is uniformly bounded and the system is globally practically uniformly exponentially stable.
\end{Lemma}

\section{Observer-Based Safety Monitoring Design}
The aim of this section is to design the interval observer gains $\underline{L}$ and $\overline{L}$, and the auxiliary neural networks $\underline{\Phi}(\underline{x},\overline{x})$ and $\overline{\Phi}(\underline{x},\overline{x})$ that satisfy Proposition \ref{Proposition 1}. In order to minimize the estimation errors, the convergence of the error system also needs to be considered. First, we introduce the design method of auxiliary neural networks $\underline{\Phi}(\underline{x},\overline{x})$ and 
$\overline{\Phi}(\underline{x},\overline{x})$ based on the neural network $\Phi(x)$ defined in (\ref{Formula 5}).

For a given neural network $\Phi$, the $l^{th}$ layer weight matrix is in the following form of
\begin{equation}
\begin{aligned}
W^{[l]}=[w^{[l]}_{i,j}]=
\begin{bmatrix}
w^{[l]}_{1,1}&w^{[l]}_{1,2}&\cdots&w^{[l]}_{1,n_{l-1}}\\
w^{[l]}_{2,1}&w^{[l]}_{2,2}&\cdots&w^{[l]}_{2,n_{l-1}}\\
\vdots&\vdots&\ddots&\vdots\\
w^{[l]}_{n_l,1}&w^{[l]}_{n_l,2}&\cdots&w^{[l]}_{n_l,n_{l-1}}\\
\end{bmatrix}\\
\end{aligned} ,
\label{Formula 10}
\end{equation}
where $w^{[l]}_{i,j}$ expresses the element in $i^{th}$ row and $j^{th}$ column. Two auxiliary weight matrices are defined as follows 
\begin{equation}
\begin{aligned}
\underline{W}^{[l]}=[\underline{w}^{[l]}_{i,j}],~\underline{w}^{[l]}_{i,j}=
\left\{
\begin{array}{ll}
w^{[l]}_{i,j}&,w^{[l]}_{i,j}<0\\
0&,w^{[l]}_{i,j}\geq0
\end{array}
\right. ,\\
\overline{W}^{[l]}=[\overline{w}^{[l]}_{i,j}],~\overline{w}^{[l]}_{i,j}=
\left\{
\begin{array}{ll}
w^{[l]}_{i,j}&,w^{[l]}_{i,j}\geq0\\
0&,w^{[l]}_{i,j}<0
\end{array}
\right. .
\end{aligned} 
\label{Formula 11}
\end{equation}

Obviously, we can get $W^{[l]}=\underline{W}^{[l]}+\overline{W}^{[l]}$. Then two auxiliary neural networks  $\underline{\Phi}(\underline{x},\overline{x}):\mathbb{R}^{2n_0}\to \mathbb{R}^{n_{L+1}}$ and $\overline{\Phi}(\underline{x},\overline{x}):\mathbb{R}^{2n_0}\to \mathbb{R}^{n_{L+1}}$ are constructed with inputs $\underline{x},\overline{x}\in\mathbb{R}^{n_0}$ in the expression of
\begin{align}
\underline{\mathcal{N}}: 
\left\{
\begin{array}{ll}
\underline{\omega}^{[0]}=\underline{x}(t)
\\
\underline{v}^{[l]}=\underline{W}^{[l]}\overline{\omega}^{[l-1]}+\overline{W}^{[l]}\underline{\omega}^{[l-1]}+b^{[l]} 
\\
\underline{\omega}^{[l]}=\phi^{[l]}(\underline{v}^{[l]}) 
\\
\underline{\Phi}(\underline{x},\overline{x})=\underline{W}^{[L+1]}\overline{\omega}^{[L]}+\overline{W}^{[L+1]}\underline{\omega}^{[L]}+b^{[L+1]}
\end{array}
\right. ,
\label{Formula 12}
\\
\overline{\mathcal{N}}: \left\{
\begin{array}{ll}
\overline{\omega}^{[0]}=\overline{x}(t) 
\\
\overline{v}^{[l]}=\underline{W}^{[l]}\underline{\omega}^{[l-1]}+\overline{W}^{[l]}\overline{\omega}^{[l-1]}+b^{[l]}  
\\
\overline{\omega}^{[l]}=\phi^{[l]}(\overline{v}^{[l]})  
\\
\overline{\Phi}(\underline{x},\overline{x})=\underline{W}^{[L+1]}\underline{\omega}^{[L]}+\overline{W}^{[L+1]}\overline{\omega}^{[L]}+b^{[L+1]}
\end{array}
\right. ,
\label{Formula 12_a}
\end{align}
where $l=1,\dots,L$.

In the case $\underline{x}\leq x\leq\overline{x}$, the following lemma proves that the auxiliary neural networks $\underline{\Phi}(\underline{x},\overline{x})$ and $\overline{\Phi}(\underline{x},\overline{x})$ identified by (\ref{Formula 12}) and (\ref{Formula 12_a}) can satisfy (\ref{proposition_2}) and (\ref{proposition_3}) in Proposition \ref{Proposition 1}, i.e. $\Phi(x)-\underline{\Phi}(\underline{x},\overline{x})\in\mathbb{R}_+^{n_{L+1}},\overline{\Phi}(\underline{x},\overline{x})-\Phi(x)\in\mathbb{R}_+^{n_{L+1}}$.

\begin{Lemma}\citep{9380552}\label{Lemma 4}
Considering the neural network $\Phi:\mathbb{R}^{n_0}\to\mathbb{R}^{n_{L+1}}$ and auxiliary neural networks  $\underline{\Phi}(\underline{x},\overline{x}):\mathbb{R}^{2n_0}\to \mathbb{R}^{n_{L+1}}$, $\overline{\Phi}(\underline{x},\overline{x}):\mathbb{R}^{2n_0}\to \mathbb{R}^{n_{L+1}}$ described by (\ref{Formula 12}) and (\ref{Formula 12_a}), the following condition
\begin{equation}
\begin{aligned}
\begin{bmatrix}
\Phi(x)-\underline{\Phi}(\underline{x},\overline{x})\\\overline{\Phi}(\underline{x},\overline{x})-\Phi(x)
\end{bmatrix}
\in\mathbb{R}_+^{2n_{L+1}}
\end{aligned} ,\label{Formula 13}
\end{equation}
holds for any $\underline{x}\leq x\leq\overline{x}$.
\end{Lemma}


The above constructed neural networks and Lemma \ref{Lemma 4} provide a method for designing the auxiliary neural networks $\underline{\Phi}(\underline{x},\overline{x})$ \wesley{and} $\overline{\Phi}(\underline{x},\overline{x})$ that meet the conditions in Proposition \ref{Proposition 1}.   Next, we need to design the observer gains $\underline{L}$ \wesley{and} $\overline{L}$ such that (\ref{proposition_1}) \wesley{and} (\ref{proposition_4}) in Proposition \ref{Proposition 1} hold and the estimation error is within an acceptable range. The nonlinear activation function makes it difficult to incorporate the above results into the convex optimization framework which is usually used for observer gain design. Inspired by the approach proposed in the literature \citep{9388885}, we can abstract the activation function by quadratic constraints.

\subsection{Quadratic Constraints on the Activation Functions}
Considering the error dynamical system (\ref{Formula 9}) and in connection with the definition of the auxiliary neural networks (\ref{Formula 12}) and (\ref{Formula 12_a}), the following results can be obtained 

\begin{align}
\Phi-\underline{\Phi}
&=W^{[L+1]}\omega^{[L]}+b^{[L+1]}-(\underline{W}^{[L+1]}\overline{\omega}^{[L]}+\overline{W}^{[L+1]}\underline{\omega}^{[L]}+b^{[L+1]})\notag\\
&=(\underline{W}^{[L+1]}+\overline{W}^{[L+1]})\omega^{[L]}-(\underline{W}^{[L+1]}\overline{\omega}^{[L]}+\overline{W}^{[L+1]}\underline{\omega}^{[L]}) \notag\\
&=\overline{W}^{[L+1]}\underline{\xi}^{[L]}-\underline{W}^{[L+1]}\overline{\xi}^{[L]} , \notag
\\
\overline{\Phi}-\Phi
&=\underline{W}^{[L+1]}\underline{\omega}^{[L]}+\overline{W}^{[L+1]}\overline{\omega}^{[L]}+b^{[L+1]}-(W^{[L+1]}\omega^{[L]}+b^{[L+1]})\notag\\
&=\underline{W}^{[L+1]}\underline{\omega}^{[L]}+\overline{W}^{[L+1]}\overline{\omega}^{[L]}-(\underline{W}^{[L+1]}+\overline{W}^{[L+1]})\omega^{[L]} \notag\\
&=-\underline{W}^{[L+1]}\underline{\xi}^{[L]}+\overline{W}^{[L+1]}\overline{\xi}^{[L]} ,\notag\\
v_\Phi-\underline{v}_\Phi
&=\begin{bmatrix}
v^{[1]}-\underline{v}^{[1]}\\
v^{[2]}-\underline{v}^{[2]}\\ 
\vdots\\
v^{[L]}-\underline{v}^{[L]}\\ 
\end{bmatrix}  
=\begin{bmatrix}
\overline{W}^{[1]}(x-\underline{x})-\underline{W}^{[1]}(\overline{x}-x)\\
\overline{W}^{[2]}\underline{\xi}^{[1]}-\underline{W}^{[2]}\overline{\xi}^{[1]}\\ 
\vdots\\
\overline{W}^{[L]}\underline{\xi}^{[L-1]}-\underline{W}^{[L]}\overline{\xi}^{[L-1]}  \notag
\end{bmatrix} ,
\\
\overline{v}_\Phi-v_\Phi
&=\begin{bmatrix}
\overline{v}^{[1]} - v^{[1]}\\ 
\overline{v}^{[2]} - v^{[2]}\\ 
\vdots\\
\overline{v}^{[L]} - v^{[L]}\\
\end{bmatrix} 
=\begin{bmatrix}
-\underline{W}^{[1]}(x-\underline{x})+\overline{W}^{[1]}(\overline{x}-x)\\
-\underline{W}^{[2]}\underline{\xi}^{[1]}+\overline{W}^{[2]}\overline{\xi}^{[1]}\\ 
\vdots\\
-\underline{W}^{[L]}\underline{\xi}^{[L-1]}+\overline{W}^{[L]}\overline{\xi}^{[L-1]}
\end{bmatrix} , \notag
\end{align}
where $\underline{\xi}^{[l]} = \omega^{[l]}-\underline{\omega}^{[l]}$ and $\overline{\xi}^{[l]} = \overline{\omega}^{[l]}-{\omega}^{[l]}$.

Furthermore, the following relationship is readily available
\begin{equation}
\begin{aligned}
\begin{bmatrix}
\Phi-\underline{\Phi}\\ \overline{\Phi}-\Phi\\v_\Phi-\underline{v}_\Phi\\ \overline{v}_\Phi-v_\Phi
\end{bmatrix}
=N\begin{bmatrix}
x-\underline{x}\\ \overline{x}-x\\ \omega_\Phi-\underline{\omega}_\Phi\\ \overline{\omega}_\Phi-\omega_\Phi
\end{bmatrix}
\end{aligned}
\label{Formula 16} ,
\end{equation}
where $N$ is defined in Table 1, and
\begin{table}
\caption{Definition of $N$ in (\ref{Formula 16})}
\centering
\captionsetup{justification=centering}
\begin{align}
&N
=\left[\begin{array}{c;{2pt/2pt}c;{2pt/2pt}c;{2pt/2pt}c}
\overline{N}_{\Phi x} & \underline{N}_{\Phi x} & \overline{N}_{\Phi \omega} & \underline{N}_{\Phi \omega}\\
 \hdashline[2pt/2pt]
\underline{N}_{\Phi x} & \overline{N}_{\Phi x} &\underline{N}_{\Phi \omega} & \overline{N}_{\Phi \omega}\\
 \hdashline[2pt/2pt]
\overline{N}_{vx} & \underline{N}_{vx} & \overline{N}_{v\omega} & \underline{N}_{v\omega}\\
 \hdashline[2pt/2pt]
\underline{N}_{vx} & \overline{N}_{vx} &\underline{N}_{v\omega} & \overline{N}_{v\omega}\\
\end{array}\right] \notag\\
&=\left[\begin{array}{c;{2pt/2pt}c;{2pt/2pt}cccc;{2pt/2pt}cccc}
0 & 0 &0 &0 & \cdots &\overline{W}^{[L+1]}&0 &0 & \cdots &-\underline{W}^{[L+1]} \\
 \hdashline[2pt/2pt]
0 & 0 &0 &0 & \cdots &-\underline{W}^{[L+1]}&0 &0 & \cdots &\overline{W}^{[L+1]} \\
 \hdashline[2pt/2pt]
\overline{W}^{[1]} & -\underline{W}^{[1]}    &0&\cdots &0&0     &0&\cdots&0&0\\
0 &0   &\overline{W}^{[2]}&\cdots &0&0     &-\underline{W}^{[2]}&\cdots &0&0\\
\vdots &\vdots   &\vdots&\ddots &\vdots&\vdots    &\vdots&\ddots &\vdots&\vdots\\
0 &0   &0&\cdots &\overline{W}^{[L]}&0   &0&\cdots &-\underline{W}^{[L]}&0\\
 \hdashline[2pt/2pt]
-\underline{W}^{[1]} & \overline{W}^{[1]}    &0&\cdots &0&0     &0&\cdots&0&0\\
0 &0   &-\underline{W}^{[2]}&\cdots &0&0     &\overline{W}^{[2]}&\cdots &0&0\\
\vdots &\vdots   &\vdots&\ddots &\vdots&\vdots    &\vdots&\ddots &\vdots&\vdots\\
0 &0   &0&\cdots &-\underline{W}^{[L]}&0   &0&\cdots &\overline{W}^{[L]}&0\\
\end{array}\right]
\notag
\end{align}
\hrule
\end{table}
\begin{equation}
\begin{aligned}
&\underline{\omega}_\Phi(t)=
\begin{bmatrix}
\underline{\omega}^{[1]}(t)\\ \vdots\\\underline{\omega}^{[L]}(t)
\end{bmatrix},
\overline{\omega}_\Phi(t)=
\begin{bmatrix}
\overline{\omega}^{[1]}(t)\\ \vdots\\ \overline{\omega}^{[L]}(t)
\end{bmatrix},
\\
&\phi(\underline{v}_\Phi)=
\begin{bmatrix}
\phi^{[1]}(\underline{v}^{[1]})\\ \vdots\\\phi^{[L]}(\underline{v}^{[L]})
\end{bmatrix} \in\mathbb{R}^{n_\Phi},
\phi(\overline{v}_\Phi)=
\begin{bmatrix}
\phi^{[1]}(\overline{v}^{[1]})\\ \vdots\\\phi^{[L]}(\overline{v}^{[L]})
\end{bmatrix}
\in\mathbb{R}^{n_\Phi},
\end{aligned}
\notag
\end{equation}
in which $n_\Phi=n_1+n_2+ \dots +n_L$.

Abstracting the activation function based on quadratic constraints (QCs) is an essential approach in the following interval observer design. Let us first define an offset local sector.

\begin{Definition}\citep{9388885}\label{Definition 2}
Suppose that given $\alpha,\beta,\hat{\underline{v}},\hat{\overline{v}},v^\ast\in\mathbb{R}$, where $\alpha\leq\beta,~\hat{\underline{v}}\leq v^\ast\leq\hat{\overline{v}}$. The activation function $\psi:\mathbb{R}\to\mathbb{R}$ satisfies the offset local sector $[\alpha,\beta]$ around the given point $(v^\ast,\psi(v^\ast))$ if
\begin{equation}
\begin{aligned}
(\Delta\psi(v)-\alpha\Delta v)(\beta\Delta v-\Delta\psi(v))\geq0,~\forall v\in[\hat{\underline{v}},\hat{\overline{v}}]
\end{aligned} ,
\label{Formula 18}
\end{equation}
where $\Delta v=v-v^\ast$ and $\Delta\psi(v)=\psi(v)-\psi(v^\ast)$.
\end{Definition}

If the function $\psi$ satisfies a local offset sector $[\alpha,\beta]$ centred at any point $(v^\ast,\psi(v^\ast))$, it means that the function $\psi$ satisfies a global offset sector $[\alpha,\beta]$. As shown in Figure \ref{FIG:1}, function $\psi(v)=tanh(v)$ satisfies the global sector bound around the point $(1,\psi(1))$ with $[\alpha,\beta]=[0,1]$. For global sector constraints, the value of $\alpha,\beta$ are independent of the chosen reference point $(v^\ast,\psi(v^\ast))$ and are only related to the chosen activation function. When the input to the function is restricted to $v\in[\hat{\underline{v}},\hat{\overline{v}}]$, the stricter offset local sector constraint will be satisfied. As shown in Figure \ref{FIG:2}, function $\psi(v)=tanh(v)$ satisfies the offset local sector bound around the point $(0,\psi(0))$ with $[\alpha,\beta]=[0.48,1]$, where $v\in[-2,2]$.

\begin{figure}[htb]
\centering
\subfloat[Global sector constraint on function $\psi(v)=tanh(v)$.\label{FIG:1}]{\includegraphics[scale=0.45]{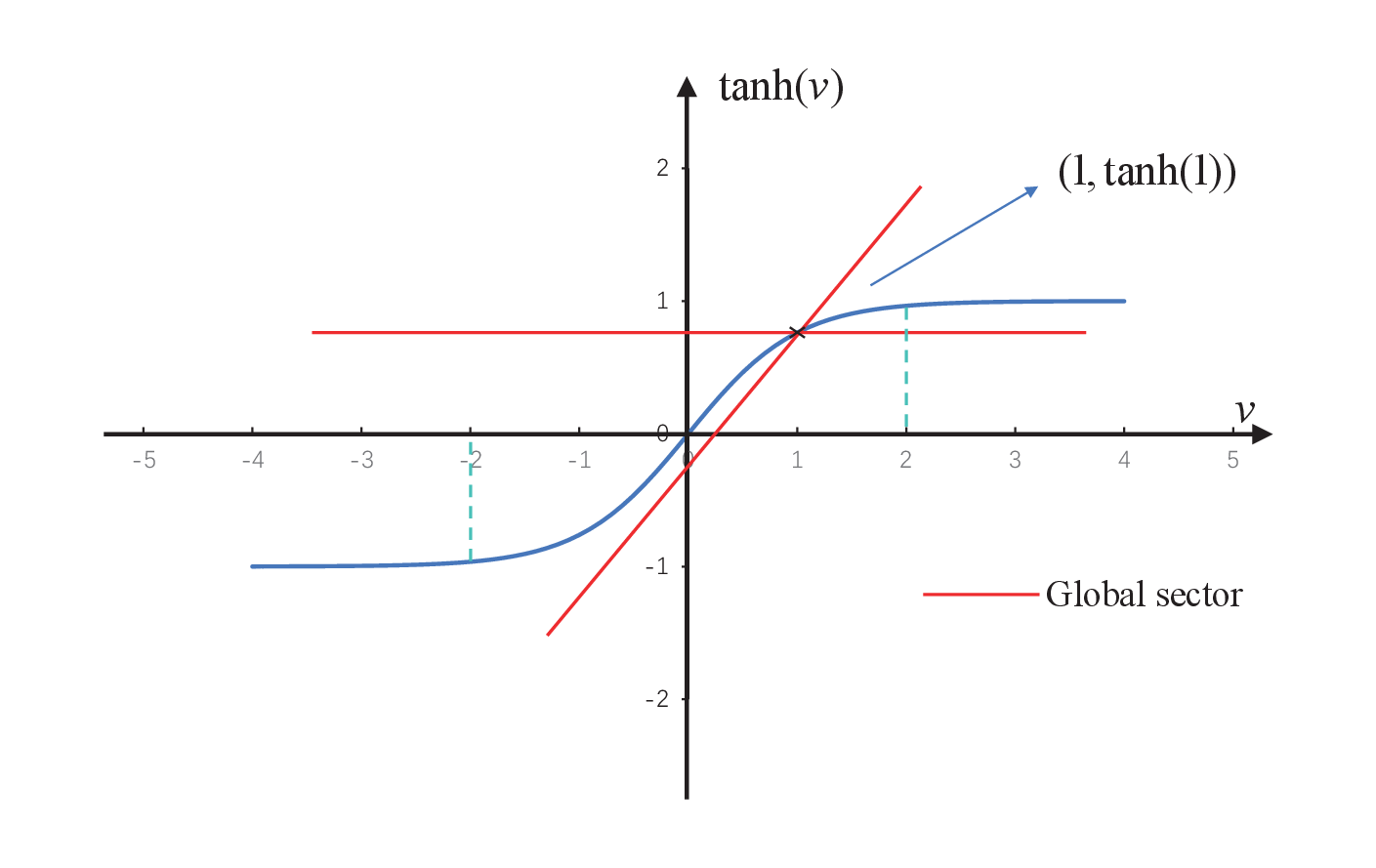}}\hspace{5pt}
\subfloat[Offset local sector constraint on function $\psi(v)=tanh(v)$.\label{FIG:2}]
{\includegraphics[scale=0.45]{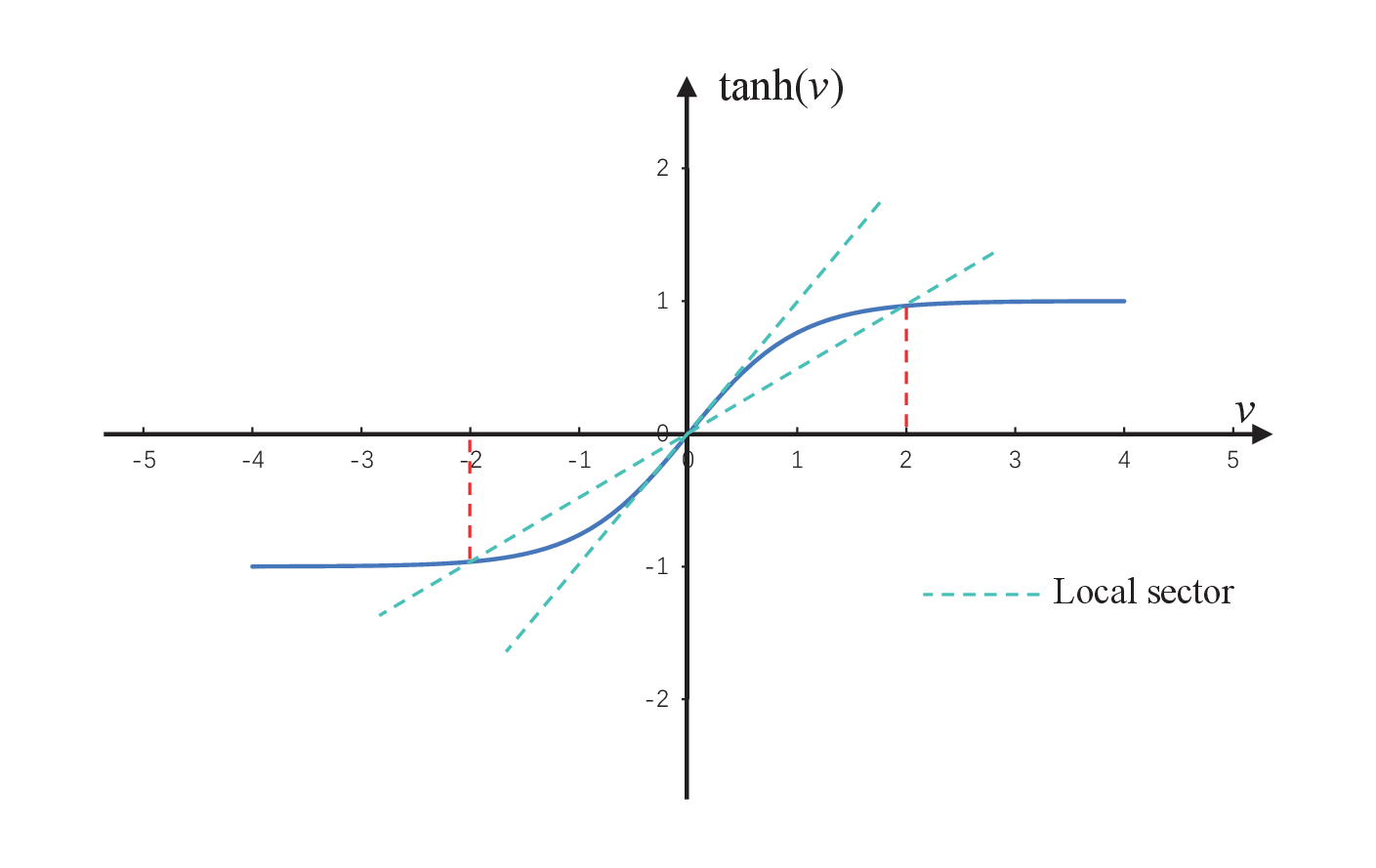}}
\caption{Two types of quadratic constraints.} \label{1}
\end{figure}

We can then convert the expression (\ref{Formula 18}) of the local offset sector into the following form
\begin{equation}
\begin{aligned}
\alpha\leq \frac{\psi(v)-\psi(v^\ast)}{v-v^\ast}\leq \beta,~
\forall v\in[\hat{\underline{v}},\hat{\overline{v}}].
\end{aligned}
\label{Formula 19}
\end{equation}

According to (\ref{Formula 19}), we can further interpret Definition \ref{Definition 2} as follows. For a function $\psi$ satisfying a local offset sector $[\alpha,\beta]$ around the point $(v^\ast,\psi(v^\ast))$, considering $\forall v\in[\hat{\underline{v}},\hat{\overline{v}}]$, the slope of the line connecting any point on the function $\psi$ to the center point $(v^\ast,\psi(v^\ast))$ is between $[\alpha,\beta]$.
The local sector constraint for a single activation function $\psi:\mathbb{R}\to\mathbb{R}$ is given above. Next, we consider the local sector constraint problem for a function formed by concatenating multiple activation functions. Considering the activation function of a series connection $\phi^{n_\Phi}:\mathbb{R}^{n_\Phi}\to\mathbb{R}^{n_\Phi}$, given $\alpha_\Phi,\beta_\Phi,\hat{\underline{v}}_\Phi,\hat{\overline{v}}_\Phi,v^\ast_\Phi \in\mathbb{R}^{n_\Phi}$, satisfying $\alpha_\Phi\leq\beta_\Phi$, $\hat{\underline{v}}_\Phi\leq v^\ast_\Phi\leq\hat{\overline{v}}_\Phi$, for the $i^{th}$ input $v_{\Phi,i}\in[\hat{\underline{v}}_{\Phi,i},\hat{\overline{v}}_{\Phi,i}],~i=1,\dots,n_\Phi$ of the function $\phi^{n_\Phi}$, we can obtain the offset sector $[\alpha_{\Phi,i},\beta_{\Phi,i}]$ either analytically or numerically. $\alpha_\Phi,\beta_\Phi$ can be obtained by stacking these local sectors, and the quadratic constraints considering the concatenation of activation functions $\phi^{n_\Phi}$ is given below.

\begin{Lemma}\citep{9388885}\label{Lemma 5}
Given $\alpha_\Phi,\beta_\Phi,\hat{\underline{v}}_\Phi,\hat{\overline{v}}_\Phi,v^\ast_\Phi \in\mathbb{R}^{n_\Phi}$, satisfying $\alpha_\Phi\leq\beta_\Phi$, $\hat{\underline{v}}_\Phi\leq v^\ast_\Phi\leq\hat{\overline{v}}_\Phi$ and $\omega^\ast_\Phi=\phi(v^\ast_\Phi)$. Suppose that the function $\phi^{n_\Phi}:\mathbb{R}^{n_\Phi}\to\mathbb{R}^{n_\Phi}$ satisfies the offset local sector $[\alpha_\Phi,\beta_\Phi]$ around the point $(v^\ast_\Phi,\psi(v^\ast_\Phi))$. Given $\lambda\geq 0$ where $\lambda\in\mathbb{R}^{n_\Phi}$, we have
\begin{equation}
\begin{aligned}
\begin{bmatrix}
v_\Phi-v^\ast_\Phi\\\omega_\Phi-\omega_\Phi^\ast
\end{bmatrix}^T
\tilde{\Psi}^T_\Phi \tilde{M}_\Phi(\lambda)\tilde{\Psi}_\Phi
\begin{bmatrix}
v_\Phi-v^\ast_\Phi\\\omega_\Phi-\omega_\Phi^\ast
\end{bmatrix}\geq 0
\end{aligned} ,\label{Formula 20}
\end{equation}
where $\omega_{\Phi}=\phi^{n_\Phi}(v)$ and
\begin{equation}
\begin{aligned}
\tilde{\Psi}_\Phi&=
\begin{bmatrix}
diag(\beta_\Phi)&-I_{n_\Phi}\\-diag(\alpha_\Phi)&I_{n_\Phi}
\end{bmatrix},\\
\tilde{M}_\Phi(\lambda)&=
\begin{bmatrix}
0_{n_\Phi}&diag(\lambda)\\diag(\lambda)&0_{n_\Phi}
\end{bmatrix}.
\end{aligned}\label{Formula 21} \notag
\end{equation}
\end{Lemma}

Lemma \ref{Lemma 5} considers the problem of quadratic constraints on the local offset sector at the level of the activation function of the entire neural network. Since our interval observers need to work properly for any input, it is necessary to consider quadratic constraints on the activation function for the global sector. Let $\underline{\hat{v}}\to-\infty,\hat{\overline{v}}\to\infty$, considering the activation function $\psi(v)=tanh(v)$, then (\ref{Formula 20}) holds if $\alpha=0,\beta=1$. According to Definition \ref{Definition 2}, $v^\ast\in\mathbb{R}$ in this case, it is feasible that $v^\ast=\underline{v}_i^{[l]}$ or $v^\ast=\overline{v}_i^{[l]}, l=1,\dots,L,i=1,\dots,n_l$. Thus, we can get
\begin{align}
\alpha\leq \frac{\psi^{[l]}(v_i^{[l]})-\psi^{[l]}(\underline{v}_i^{[l]})}{v_i^{[l]}-\underline{v}_i^{[l]}}\leq \beta,
\label{Formula 22}
\\
\alpha\leq \frac{\psi^{[l]}(\overline{v}_i^{[l]})-\psi^{[l]}(v_i^{[l]})}{\overline{v}_i^{[l]}-v_i^{[l]}}\leq \beta .
\label{Formula 22_a}
\end{align}

Similarly, considering the case for the global sector, we can obtain the global sector quadratic constraints on the activation function applied to the neural network (\ref{Formula 5}) and the auxiliary neural networks (\ref{Formula 12}) and (\ref{Formula 12_a}) in the error dynamical system (\ref{Formula 9}) as follows.

\begin{Theorem}\label{Theorem 1}
Given $\alpha_\Phi,\beta_\Phi\in\mathbb{R}^{n_\Phi}$ and existing $\underline{v}_\Phi,\overline{v}_\Phi,v_\Phi \in\mathbb{R}^{n_\Phi}$, satisfying $\alpha_\Phi\leq\beta_\Phi$, $\underline{v}_\Phi\leq v_\Phi\leq\overline{v}_\Phi$ and $\omega_\Phi=\phi^{n_\Phi}(v_\Phi)$. 
Consider the definition of the neural network (\ref{Formula 5}) and auxiliary neural networks (\ref{Formula 12}) and (\ref{Formula 12_a}), for exactly the same activation function of the concatenation $\phi^{n_\Phi}=[\psi,\dots,\psi]:\mathbb{R}^{n_\Phi}\to\mathbb{R}^{n_\Phi}$.  Given $\lambda\geq 0$ where $\lambda\in\mathbb{R}^{n_\Phi}$, we have
\begin{equation}
\begin{aligned}
\Pi=\begin{bmatrix}
v_\Phi-\underline{v}_\Phi\\\overline{v}_\Phi-v_\Phi\\\omega_\Phi-\underline{\omega}_\Phi\\\overline{\omega}_\Phi-\omega_\Phi
\end{bmatrix}^T
\Psi^T_\Phi M_\Phi(\lambda)\Psi_\Phi
\begin{bmatrix}
v_\Phi-\underline{v}_\Phi\\\overline{v}_\Phi-v_\Phi\\\omega_\Phi-\underline{\omega}_\Phi\\\overline{\omega}_\Phi-\omega_\Phi
\end{bmatrix}\geq 0
\end{aligned} ,\label{Formula 23}
\end{equation}
where 
\begin{equation}
\begin{aligned}
\Psi_\Phi&=
\begin{bmatrix}
diag(\beta_\Phi)&0_{n_\Phi}&-I_{n_\Phi}&0_{n_\Phi}\\
0_{n_\Phi}&diag(\beta_\Phi)&0_{n_\Phi}&-I_{n_\Phi}\\
-diag(\alpha_\Phi)&0_{n_\Phi}&I_{n_\Phi}&0_{n_\Phi}\\
0_{n_\Phi}&-diag(\alpha_\Phi)&0_{n_\Phi}&I_{n_\Phi}\\
\end{bmatrix} ,\\
M_\Phi(\lambda)&=
\begin{bmatrix}
0_{n_\Phi}&0_{n_\Phi}&diag(\lambda)&0_{n_\Phi}\\
0_{n_\Phi}&0_{n_\Phi}&0_{n_\Phi}&diag(\lambda)\\
diag(\lambda)&0_{n_\Phi}&0_{n_\Phi}&0_{n_\Phi}\\
0_{n_\Phi}&diag(\lambda)&0_{n_\Phi}&0_{n_\Phi} \\
\end{bmatrix},
\end{aligned} \label{Formula 24} \notag
\end{equation}
and $\alpha_\Phi=[\alpha,\dots,\alpha]^T,\beta_\Phi=[\beta,\dots,\beta]^T\in\mathbb{R}^{n_\Phi}$, which can be obtained by analyzing the 
global sector constraints on a single activation function.
\begin{proof}
According to (\ref{Formula 23}), for any $\underline{v}_\Phi\leq v_\Phi\leq\overline{v}_\Phi$, one can obtain
\begin{equation}
\begin{aligned}
&\begin{bmatrix}
v_\Phi-\underline{v}_\Phi\\\overline{v}_\Phi-v_\Phi\\\omega_\Phi-\underline{\omega}_\Phi\\\overline{\omega}_\Phi-\omega_\Phi
\end{bmatrix}^T
\Psi^T_\Phi M_\Phi(\lambda)\Psi_\Phi
\begin{bmatrix}
v_\Phi-\underline{v}_\Phi\\\overline{v}_\Phi-v_\Phi\\\omega_\Phi-\underline{\omega}_\Phi\\\overline{\omega}_\Phi-\omega_\Phi
\end{bmatrix}\\
=&
\begin{bmatrix}
\Delta\underline{v}_\Phi\\\Delta\overline{v}_\Phi\\\Delta\underline{\omega}_\Phi\\\Delta\overline{\omega}_\Phi
\end{bmatrix}^T
\Psi^T_\Phi M_\Phi(\lambda)\Psi_\Phi
\begin{bmatrix}
\Delta\underline{v}_\Phi\\\Delta\overline{v}_\Phi\\\Delta\underline{\omega}_\Phi\\\Delta\overline{\omega}_\Phi
\end{bmatrix}\\
=&\sum_{l=1}^L\sum_{i=1}^{n_l}\lambda_i^{[l]}(\Delta\underline{\omega}_i^{[l]}-\alpha\Delta\underline{v}_i^{[l]})(\beta\Delta\underline{v}_i^{[l]}-\Delta\underline{\omega}_i^{[l]})\\
&~~~~+\sum_{l=1}^L\sum_{i=1}^{n_l}\lambda_i^{[l]}(\Delta\overline{\omega}_i^{[l]}-\alpha\Delta\overline{v}_i^{[l]})(\beta\Delta\overline{v}_i^{[l]}-\Delta\overline{\omega}_i^{[l]}) ,
 \end{aligned} 
\notag
\end{equation}
where $\Delta\underline{v}_i^{[l]}=v_i^{[l]}-\underline{v}_i^{[l]}$, $\Delta\overline{v}_i^{[l]}=\overline{v}_i^{[l]}-v_i^{[l]},\Delta\underline{\omega}_i^{[l]}=\omega_i^{[l]}-\underline{\omega}_i^{[l]}$, $\Delta\overline{\omega}_i^{[l]}=\overline{\omega}_i^{[l]}-\omega_i^{[l]}$. Using (\ref{Formula 22}) and (\ref{Formula 22_a}), it is easy to see that each term in the equation is non-negative in the case of $\lambda_i^{[l]}\geq0$.
Thus the proof is complete.
\end{proof}
\end{Theorem}

\begin{Remark}\label{Remark 4}
For the description of the quadratic constraints on  activation functions, our approach uses an extension based on the description of the local constraint in Definition \ref{Definition 2} to obtain the global constraint needed for the subsequent proof. For example, for the activation function $\psi(v)=tanh(v)$ mentioned in Figure \ref{FIG:1}, we have $\alpha=0$, $\beta=1$. Therefore, we can obtain $\alpha_\Phi=[0,\dots,0]^T,~\beta_\Phi=[1,\dots,1]^T$.
Definition \ref{Definition 2} can be generalized from local constraint to global constraint, and a more detailed discussion on the quadratic constraints can be found in the paper \citep{9388885}.
\end{Remark}

\subsection{Design of Interval Observer}
This section uses the Lyapunov stability theory and the global sector constraint of the activation function given in Theorem \ref{Theorem 1} to obtain the tractable linear matrix inequality (LMI) and  conditions that ensure the error dynamical system (\ref{Formula 9}) a positive system and practically stable.

\begin{Theorem}\label{Theorem 2}
Considering the error dynamical system (\ref{Formula 9}) with the definition of neural network (\ref{Formula 5}) and the definition of auxiliary neural networks (\ref{Formula 12}) and (\ref{Formula 12_a}), and nonlinear function under Assumptions \ref{Assumption 1}, \ref{Assumption 2}, if there exist diagonal matrix $Q\succ0$, diagonal matrix $S$ and block diagonal matrix $M$, real number $k_1>0$, such that
\begin{align}
\begin{bmatrix}
\varGamma_1
&\varGamma_2
&Q\tilde{B}_\Phi&Q\tilde{B}_u&Q\\
\varGamma_2^T
&\varGamma_3
&0&0&0\\
\tilde{B}_\Phi^TQ&0&-I&0&0\\
\tilde{B}_u^TQ&0&0&-I&0\\
Q&0&0&0&-I\\
\end{bmatrix}
\prec 0 , \label{Formula 25}
\\
Q\tilde{A}-M\tilde{C}+S\geq0 ,
\label{Formula 26}
\end{align}
where 
$\varGamma_1=Q\tilde{A}-M\tilde{C}+\tilde{A}^TQ-\tilde{C}^TM^T+k_1I+\tilde{N}_{\Phi x}^T\tilde{N}_{\Phi x}+\tilde{N}^T_{vx}F_{\alpha\beta}\tilde{N}_{vx}$, $
\varGamma_2=\tilde{N}_{\Phi x}^T \tilde{N}_{\Phi \omega}+\tilde{N}_{v x}^TF_{\alpha\beta}\tilde{N}_{v \omega}+\tilde{N}_{vx}^TF_{\alpha+\beta}$, and $
\varGamma_3=\tilde{N}_{\Phi \omega}^T\tilde{N}_{\Phi \omega}+\tilde{N}_{v \omega}^TF_{\alpha\beta}\tilde{N}_{v \omega}+F_{\alpha+\beta}\tilde{N}_{v \omega}+\tilde{N}_{v \omega}^TF_{\alpha+\beta}+F_\lambda$ in which 
\begin{equation}
\begin{aligned}
&\tilde{A}=
\begin{bmatrix}
A&0\\0&A
\end{bmatrix},~
\tilde{B}_\Phi=
\begin{bmatrix}
B_\Phi&0\\0&B_\Phi
\end{bmatrix},~
\tilde{B}_u=
\begin{bmatrix}
B_u&0\\0&B_u
\end{bmatrix} ,
\\
&\tilde{C}=
\begin{bmatrix}
C&0\\0&C
\end{bmatrix},~
\tilde{L}=
\begin{bmatrix}
\underline{L}&0\\0&\overline{L}
\end{bmatrix},~
M=
\begin{bmatrix}
M_1&0\\0&M_2
\end{bmatrix} ,\\
&\begin{bmatrix}
\tilde{N}_{\Phi x}&\tilde{N}_{\Phi \omega}\\\tilde{N}_{v x}&\tilde{N}_{v\omega}
\end{bmatrix}=
\left[\begin{array}{cc;{2pt/2pt}cc}
\overline{N}_{\Phi x}&\underline{N}_{\Phi x}&\overline{N}_{\Phi \omega}&\underline{N}_{\Phi \omega}\\
\underline{N}_{\Phi x}&\overline{N}_{\Phi x}&\underline{N}_{\Phi \omega}&\overline{N}_{\Phi \omega}\\
\hdashline[2pt/2pt]
\overline{N}_{v x}&\underline{N}_{v x}&\overline{N}_{v \omega}&\underline{N}_{v \omega}\\
\underline{N}_{v x}&\overline{N}_{v x}&\underline{N}_{v \omega}&\overline{N}_{v \omega}\\
\end{array}\right] ,\\
&\Psi^T_\Phi M_\Phi(\lambda)\Psi_\Phi=\begin{bmatrix}
F_{\alpha\beta}&F_{\alpha+\beta}\\F_{\alpha+\beta}&F_{\lambda}
\end{bmatrix}=
\left[\begin{array}{cc;{2pt/2pt}cc}
\lambda_1&0_{n_\Phi}&\lambda_2&0_{n_\Phi}\\
0_{n_\Phi}&\lambda_1&0_{n_\Phi}&\lambda_2\\
\hdashline[2pt/2pt]
\lambda_2&0_{n_\Phi}&\lambda_3&0_{n_\Phi}\\
0_{n_\Phi}&\lambda_2&0_{n_\Phi}&\lambda_3\\
\end{array}\right] ,\\
&\lambda_1=-2\alpha\beta diag(\lambda),~
\lambda_2=(\alpha+\beta) diag(\lambda),~
\lambda_3=-2diag(\lambda) ,\\
\end{aligned}\notag
\end{equation}
and $\alpha,\beta\in\mathbb{R}$ are the exact values determined by the chosen activation function, and $k_1=3max\{(\underline{a}_1^2+\overline{a}_1^2),(\underline{a}_2^2+\overline{a}_2^2)\}$, which can be calculated by Assumption \ref{Assumption 2},  then the error dynamical system (\ref{Formula 9}) is a practically stable and positive system. The system (\ref{Formula 8}) is an interval observer of the nonlinear system (\ref{Formula 2}) and the observer gains \wesley{matrices} $\underline{L}$ \wesley{and} $ \overline{L}$ can be obtained by $\tilde{L} =Q^{-1}M$.
\begin{proof}
Since (\ref{Formula 25}) holds, let $E=\tilde{A}-\tilde{L}\tilde{C}$ and according to the relation $\tilde{L}=Q^{-1}M$, (\ref{Formula 25}) can be rewritten as
\begin{equation}
\begin{aligned}
\begin{bmatrix}
\varGamma_4
&\varGamma_2
&Q\tilde{B}_\Phi&Q\tilde{B}_u&Q\\
\varGamma_2^T
&\varGamma_3&0&0&0\\
\tilde{B}_\Phi^TQ&0&-I&0&0\\
\tilde{B}_u^TQ&0&0&-I&0\\
Q&0&0&0&-I\\
\end{bmatrix}
\prec0,\\
\end{aligned}\label{Formula 27}
\end{equation}
where $\varGamma_4=QE+E^TQ+k_1I+\tilde{N}_{\Phi x}^T\tilde{N}_{\Phi x}+\tilde{N}^T_{vx}F_{\alpha\beta}\tilde{N}_{vx}$. 

Using the Schur complement equivalence, (\ref{Formula 27}) can be equivalent to be
\begin{equation}
\begin{aligned}
\begin{bmatrix}
\varGamma_5
&\varGamma_2\\
\varGamma_2^T
&\varGamma_3\\
\end{bmatrix}
\prec0,\\
\end{aligned}\label{Formula 28}
\end{equation}
where $\varGamma_5=QE+E^TQ+k_1I+Q\tilde{B}_\Phi\tilde{B}_\Phi^TQ+Q\tilde{B}_u\tilde{B}_u^TQ+QQ+\tilde{N}_{\Phi x}^T\tilde{N}_{\Phi x}+\tilde{N}^T_{vx}F_{\alpha\beta}\tilde{N}_{vx}$.

Split (\ref{Formula 28}) into two matrices $J_1$ \wesley{and} $J_2$ such that
\begin{equation}
\begin{aligned}
J_1&=
\begin{bmatrix}
\varGamma_6+\tilde{N}_{\Phi x}^T\tilde{N}_{\Phi x}
&\tilde{N}_{\Phi x}^T \tilde{N}_{\Phi \omega}\\
\tilde{N}_{\Phi \omega}^T\tilde{N}_{\Phi x}
&\tilde{N}_{\Phi \omega}^T\tilde{N}_{\Phi \omega}\\
\end{bmatrix},\\
J_2&=
\begin{bmatrix}
\tilde{N}^T_{vx}F_{\alpha\beta}\tilde{N}_{vx}
&\varGamma_7\\
\tilde{N}_{v \omega}^TF_{\alpha\beta}\tilde{N}_{vx}+F_{\alpha+\beta}\tilde{N}_{vx}
&\varGamma_8\\
\end{bmatrix}.\\
\end{aligned}\label{Formula 29}
\end{equation}
where $\varGamma_6=QE+E^TQ+k_1I+Q\tilde{B}_\Phi\tilde{B}_\Phi^TQ+Q\tilde{B}_u\tilde{B}_u^TQ+QQ$, $\varGamma_7=\tilde{N}_{v x}^TF_{\alpha\beta}\tilde{N}_{v \omega}+\tilde{N}_{vx}^TF_{\alpha+\beta}$, $\varGamma_8=\tilde{N}_{v \omega}^TF_{\alpha\beta}\tilde{N}_{v \omega}+F_{\alpha+\beta}\tilde{N}_{v \omega}+\tilde{N}_{v \omega}^TF_{\alpha+\beta}+F_\lambda$.

Clearly, $J_1+J_2\prec0$ and the following relation holds
\begin{equation}
\begin{aligned}
J_1&=
\begin{bmatrix}
I&0\\\tilde{N}_{\Phi x}&\tilde{N}_{\Phi \omega}
\end{bmatrix}^T
\begin{bmatrix}
\varGamma_6&0\\
0&I\\
\end{bmatrix}
\begin{bmatrix}
I&0\\\tilde{N}_{\Phi x}&\tilde{N}_{\Phi \omega}
\end{bmatrix},\\
J_2&=
\begin{bmatrix}
\tilde{N}_{v x}&\tilde{N}_{v \omega}\\0&I
\end{bmatrix}^T
\begin{bmatrix}
F_{\alpha\beta}&F_{\alpha+\beta}\\F_{\alpha+\beta}&F_{\lambda}
\end{bmatrix}
\begin{bmatrix}
\tilde{N}_{v x}&\tilde{N}_{v \omega}\\0&I
\end{bmatrix}.\\
\end{aligned}\label{Formula 30}
\end{equation}

From (\ref{Formula 16}), multiplying of the matrix inequality $J_1+J_2\prec0$ left and right by $[(x-\underline{x})^T,(\overline{x}-x)^T,(\omega_\Phi-\underline{\omega}_\Phi)^T,(\overline{\omega}_\Phi-\omega_\Phi)^T]$ and its transpose, we have
\begin{equation}
\begin{aligned}
\begin{bmatrix}
x-\underline{x}\\ \overline{x}-x\\\Phi-\underline{\Phi}\\ \overline{\Phi}-\Phi
\end{bmatrix}^T
\begin{bmatrix}
\varGamma_6&0\\
0&I\\
\end{bmatrix}
\begin{bmatrix}
x-\underline{x}\\ \overline{x}-x\\\Phi-\underline{\Phi}\\ \overline{\Phi}-\Phi
\end{bmatrix}
+\Pi<0 .
\end{aligned}
\label{Formula 31}
\end{equation}

Combining (\ref{Formula 23}) in Theorem \ref{Theorem 1}, we can conclude that
\begin{equation}
\begin{aligned}
\begin{bmatrix}
x-\underline{x}\\ \overline{x}-x\\\Phi-\underline{\Phi}\\ \overline{\Phi}-\Phi
\end{bmatrix}^T
\begin{bmatrix}
\varGamma_6&0\\
0&I\\
\end{bmatrix}
\begin{bmatrix}
x-\underline{x}\\ \overline{x}-x\\\Phi-\underline{\Phi}\\ \overline{\Phi}-\Phi
\end{bmatrix}
<0.
\end{aligned}
\label{Formula 32}
\end{equation}

Then we consider the error dynamical system (\ref{Formula 9}) and rewrite it as
\begin{equation}
\begin{aligned}
\dot{\tilde{e}}&=
\begin{bmatrix}
A-\underline{L}C&0\\0&A-\overline{L}C
\end{bmatrix}
\begin{bmatrix}
\underline{e}\\\overline{e}
\end{bmatrix}+
\begin{bmatrix}
B_\Phi&0\\0&B_\Phi
\end{bmatrix}
\begin{bmatrix}
\Phi-\underline{\Phi}\\\overline{\Phi}-\Phi
\end{bmatrix}\\
&~~~~+\begin{bmatrix}
B_u&0\\0&B_u
\end{bmatrix}
\begin{bmatrix}
u-\underline{u}\\\overline{u}-u
\end{bmatrix}+
\begin{bmatrix}
f(x)-\underline{f}(\underline{x},\overline{x})\\\overline{f}(\underline{x,\overline{x}})-f(x)
\end{bmatrix}\\
&=E\tilde{e}+\tilde{B}_\Phi\Delta \Phi+\tilde{B}_u\Delta u+\tilde{f}.
\end{aligned}
\label{Formula 33}
\end{equation}

To prove the stability of error dynamical system (\ref{Formula 9}), let us consider a Lyapunov function $V(t)=\tilde{e}^TQ\tilde{e}$, whose time derivative takes the form:
\begin{equation}
\begin{aligned}
\dot{V}(t)&=(\dot{\tilde{e}}^TQ\tilde{e}+\tilde{e}^TQ\dot{\tilde{e}})\\
&=(E\tilde{e}+\tilde{B}_\Phi\Delta \Phi+\tilde{B}_u\Delta u+\tilde{f})^TQ\tilde{e}+\tilde{e}^TQ(E\tilde{e}+\tilde{B}_\Phi\Delta \Phi+\tilde{B}_u\Delta u+\tilde{f})\\
&=\tilde{e}^T(QE+E^TQ)\tilde{e}+2\tilde{e}^TQ\tilde{B}_\Phi\Delta\Phi+2\tilde{e}^TQ\tilde{B}_u\Delta u+2\tilde{f}^TQ\tilde{e}.
\end{aligned}
\label{Formula 34}
\end{equation}

The following inequalities are introduced
\begin{equation}
\begin{aligned}
2\tilde{e}^TQ\tilde{B}_\Phi\Delta\Phi&\leq\tilde{e}^TQ\tilde{B}_\Phi\tilde{B}_\Phi^TQ\tilde{e}+\Delta\Phi^T\Delta\Phi,\\
2\tilde{e}^TQ\tilde{B}_u\Delta u&\leq\tilde{e}^TQ\tilde{B}_u\tilde{B}_u^TQ\tilde{e}+\Delta u^T\Delta u,\\
2\tilde{f}^TQ\tilde{e}&\leq\tilde{e}^TQQ\tilde{e}+\tilde{f}^T\tilde{f}.\\
\end{aligned}
\notag
\end{equation}

Under Assumption \ref{Assumption 2}, it implies that

\begin{equation}
\begin{aligned}
\tilde{f}^T\tilde{f}&=
(\underline{a}_1\underline{e}+\underline{a}_2\overline{e}+\underline{\rho})^T(\underline{a}_1\underline{e}+\underline{a}_2\overline{e}+\underline{\rho})+
(\overline{a}_1\underline{e}+\overline{a}_2\overline{e}+\overline{\rho})^T(\overline{a}_1\underline{e}+\overline{a}_2\overline{e}+\overline{\rho})\\
&=\Vert\underline{a}_1\underline{e}+\underline{a}_2\overline{e}+\underline{\rho}\Vert^2+\Vert\overline{a}_1\underline{e}+\overline{a}_2\overline{e}+\overline{\rho}\Vert^2\\
&\leq3(\underline{a}_1^2\Vert\underline{e}\Vert^2+\underline{a}_2^2\Vert\overline{e}\Vert^2+\Vert\underline{\rho}\Vert^2)+3(\overline{a}_1^2\Vert\underline{e}\Vert^2+\overline{a}_2^2\Vert\overline{e}\Vert^2+\Vert\overline{\rho}\Vert^2)\\
&=\begin{bmatrix}
\underline{e}\\\overline{e}
\end{bmatrix}^T
\begin{bmatrix}
3(\underline{a}_1^2+\overline{a}_1^2)&0\\0&3(\underline{a}_2^2+\overline{a}_2^2)
\end{bmatrix}
\begin{bmatrix}
\underline{e}\\\overline{e}
\end{bmatrix}
+3(\Vert\underline{\rho}\Vert^2+\Vert\overline{\rho}\Vert^2)\\
&\leq\tilde{e}^T k_1I\tilde{e}+3k_2 ,
\end{aligned}
\label{Formula 35} \notag
\end{equation}
where $k_1=3max\{(\underline{a}_1^2+\overline{a}_1^2),(\underline{a}_2^2+\overline{a}_2^2)\},k_2=\Vert\underline{\rho}\Vert^2+\Vert\overline{\rho}\Vert^2$.

Therefore, (\ref{Formula 34}) implies that
\begin{equation}
\begin{aligned}
\dot{V}(t)
&\leq\tilde{e}^T\varGamma_7\tilde{e}+\Delta\Phi^T\Delta\Phi+\Delta u^T\Delta u+\tilde{f}^T\tilde{f}\\
&\leq\tilde{e}^T\varGamma_7\tilde{e}+\Delta\Phi^T\Delta\Phi+\tilde{e}^T k_1I\tilde{e}+\Delta u^T\Delta u+3k_2\\
&=\tilde{e}^T(\varGamma_7+k_1I)\tilde{e}+\Delta\Phi^T\Delta\Phi+\Delta u^T\Delta u+3k_2\\
&=\tilde{e}^T\varGamma_6\tilde{e}+\Delta\Phi^T\Delta\Phi+\Delta u^T\Delta u+3k_2,\\
\end{aligned}
\label{Formula 36} \notag
\end{equation}
where $\varGamma_7=QE+E^TQ+Q\tilde{B}_\Phi\tilde{B}_\Phi^TQ+Q\tilde{B}_u\tilde{B}_u^TQ+QQ$.

According to (\ref{Formula 32}), we can get the following condition
\begin{equation}
\begin{aligned}
\begin{bmatrix}
\tilde{e}\\ \Delta\Phi
\end{bmatrix}^T
\begin{bmatrix}
\varGamma_6&0\\
0&I\\
\end{bmatrix}
\begin{bmatrix}
\tilde{e}\\ \Delta\Phi
\end{bmatrix}
<0,
\end{aligned}
\label{Formula 37}
\end{equation}
from which it can be inferred that
\begin{equation}
\begin{aligned}
\tilde{e}^T\varGamma_6\tilde{e}+\Delta\Phi^T\Delta\Phi<0.
\end{aligned}
\label{Formula 38}
\end{equation}

According to the conditions (\ref{Formula 38}) derived above, in any case, there must be a real number $\varepsilon>0$, such that the following equation holds
\begin{equation}
\begin{aligned}
\tilde{e}^T(\varGamma_6+\varepsilon Q)\tilde{e}+\Delta\Phi^T\Delta\Phi<0.
\end{aligned}
\label{Formula 39}
\end{equation}
which can be rewritten to
\begin{equation}
\begin{aligned}
\tilde{e}^T\varGamma_6\tilde{e}+\Delta\Phi^T\Delta\Phi<-\varepsilon\tilde{e}^T Q\tilde{e}.
\end{aligned}
\label{Formula 40}
\end{equation}

Moreover, based on (\ref{Formula 36}), we can conclude that
\begin{equation}
\begin{aligned}
\dot{V}(t)&\leq-\varepsilon V+c_2 ,\\
\end{aligned}
\label{Formula 41}
\end{equation}
where $c_2\in\mathbb{R}_+$ and $c_2\geq\Delta u^T\Delta u+3k_2$.
According to Lemma \ref{Lemma 2}, the error dynamical system (\ref{Formula 9}) is practically stable.

Adding or subtracting $S$ does not affect the Metzler property of the expression because $S$ is a diagonal matrix. Thus $Q\tilde{A}-M\tilde{C}$ is Metzler considering matrix inequality (\ref{Formula 26}). Based on $M=Q\tilde{L}$, then $Q\tilde{A}-Q\tilde{L}\tilde{C}$ is Metzler. Multiplying the diagonal matrix $Q$ will not change the Metzler features based on Lemma \ref{Lemma 3}, so $\tilde{A}-\tilde{L}\tilde{C}$ is Metzler. Thus the proof is complete.

\end{proof}
\end{Theorem}

\begin{Remark}\label{Remark 5}
It is worth noting that considering the setting $\underline{x}(0)\leq x(0)\leq \overline{x}(0)$ mentioned in Assumption \ref{Assumption 1}, it is possible the left side of (\ref{Formula 31}) is equal to 0 when $t=0$. Since $\underline{u}<u<\overline{u}$, according to the theory of positive systems, the state variable $\underline{e}=\overline{e}=0$ in system (\ref{Formula 9}) when and only when $t = 0$. This means that the case $\underline{x}= x = \overline{x}$ exists only when $t = 0$. Here we consider the fact that (\ref{Formula 31}) does not hold only at this moment $t = 0$ \wesley{and} does not have an impact on the correctness of Theorem \ref{Theorem 2}. The main reason why the case $\underline{x}(0) = x(0) =\overline{x}(0)$ is retained in the assumption on the initial value of the system state is to \wesley{minimize} the usage restrictions of the proposed method.
\end{Remark}

\section{Application to Lateral Vehicle Control Systems}

In this section, the developed run-time safety monitor design methodology is applied to the lateral vehicle control system to evaluate the correctness and applicability of the proposed methodology. 
The National Highway Transportation Safety Administration (NHTSA) has identified lane departures as the leading cause of rollovers in sport utility vehicles (SUVs) and light trucks (http://www.nhtsa.gov). Lateral vehicle control is an important approach to resolving lane departure accidents and has been heavily researched in industry and academia.
Lateral vehicle control means that the vehicle collects road and environmental information via sensors such as magnetic materials, vision systems\wesley{,} or GPS to obtain the vehicle's position relative to the desired path. Control commands are then issued to the vehicle based on a control strategy. The control process can be summarized into two parts: detection and reaction. The detection device evaluates the position of the vehicle relative to the road in real time and determines whether a road deviation has occurred. Once a deviation is detected, the controller issues a warning to the driver and/or intervenes in the vehicle.

\begin{figure}[htb]
\centering
		\includegraphics[scale=0.4]{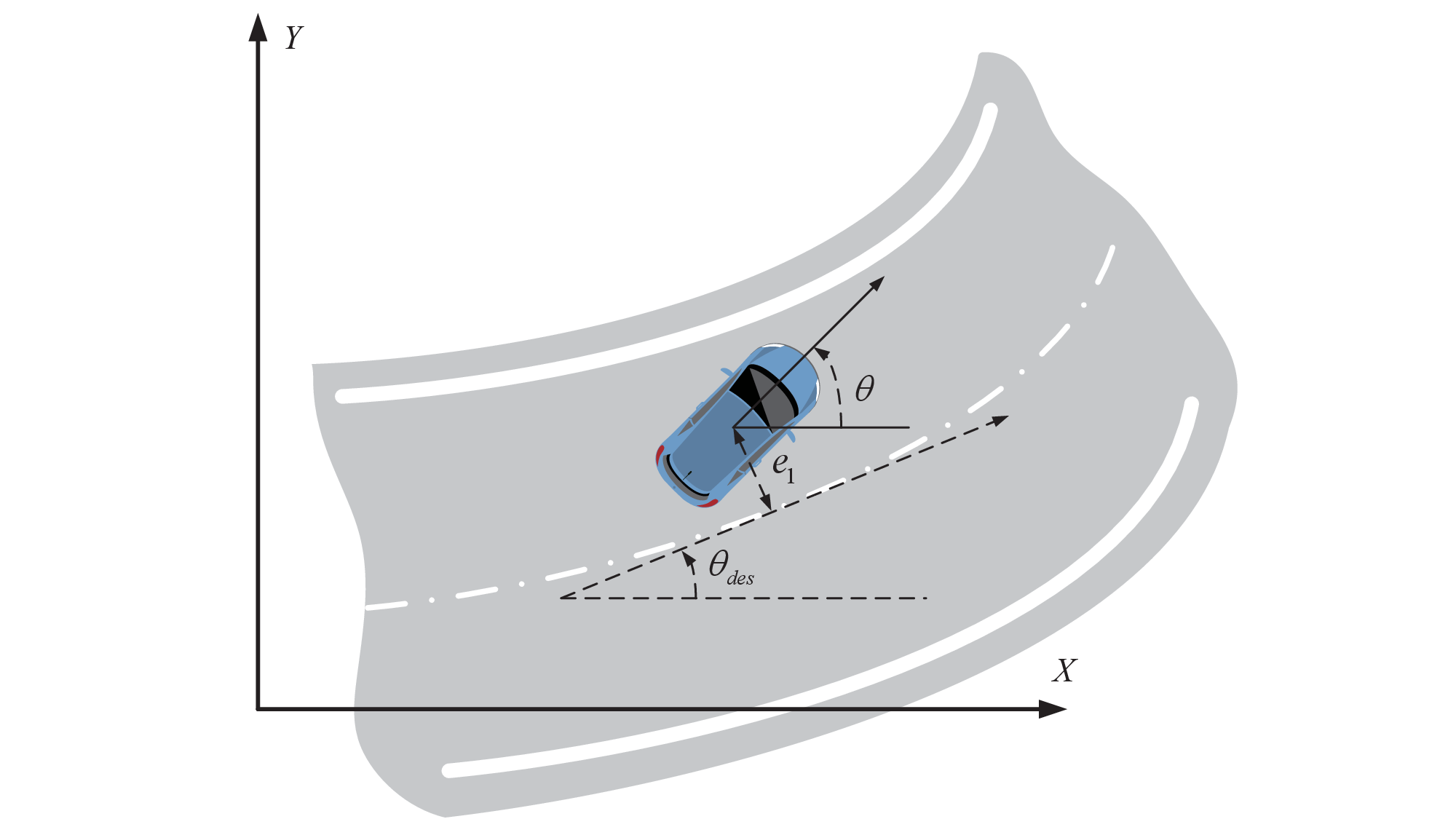}
	\caption{Illustration of Lateral Vehicle Control System.}
	\label{FIG:5}
\end{figure}

In the following example, we consider a “bicycle” model of a vehicle with two degrees of freedom, the lateral position $Y$ of the vehicle, and the yaw angle $\theta$ of the vehicle, as shown in Figure \ref{FIG:5}.
This system is under the control of a neural network controller, which serves to provide intervention when the vehicle leaves the road center line. Consider the lateral vehicle control system from \citep{rajamani2011vehicle}:
\begin{equation}
\begin{aligned}
\mathcal{L}_{eg}:\left\{
\begin{array}{ll}
\dot{x}=Ax+B_\Phi \Phi(x)+B_u u\\
y=Cx
\end{array}
\right.
\end{aligned} ,
\label{Formula 42}
\end{equation}
where system matrices $A$, $B_\Phi$, $B_u$, $C$, and input $u$ are defined by
\begin{equation}
\begin{aligned}
&A=\begin{bmatrix}
0&1&0&0\\
0&-\frac{2C_{af}+2C_{ar}}{mV_x}&\frac{2C_{af}+2C_{ar}}{m}&\frac{-2C_{af}l_f+2C_{ar}l_r}{mV_x}\\
0&0&0&1\\
0&-\frac{2C_{af}l_f-2C_{ar}l_r}{I_zV_x}&\frac{2C_{af}l_f-2C_{ar}l_r}{I_z}&-\frac{2C_{af}l_f^2+2C_{ar}l_r^2}{I_zV_x}\\
\end{bmatrix},
~B_\Phi=\begin{bmatrix}
0\\\frac{2C_{af}}{m}\\0\\\frac{2C_{af}l_f}{I_z}
\end{bmatrix},\\
&B_u=
\begin{bmatrix}
1&1&0&0\\
0&1&0&0\\
0&0&1&0\\
0&0&0&1
\end{bmatrix},
C=
\begin{bmatrix}
0&1&0&0\\
1&0&0&0\\
0&0&1&1\\
0&0&0&1
\end{bmatrix}
,~u=\begin{bmatrix}
0\\-\frac{2C_{af}l_f-2C_{ar}l_r}{mV_x}-V_x\\0\\-\frac{2C_{af}l_f^2+2C_{ar}l_r^2}{I_zV_x}
\end{bmatrix}
\dot{\theta}_{des}.
\end{aligned}
\notag
\end{equation}

Let $x=[e_1,\dot{e}_1,e_2,\dot{e}_2]^T$ denote the state of the system (\ref{Formula 42}). $e_1$ is the distance of the c.g. (center of gravity) of the vehicle from the centreline of the lane (m). $e_2$ is the orientation error of the vehicle with respect to the road (rad), which can be obtained by the equation $e_2=\theta-\theta_{des}$.
$\theta$ is called the heading angle of the vehicle, and $\theta_{des}$ is the desired orientation of the vehicle, with respect to the global X-axis (rad). $\dot{\theta}_{des}=\frac{V_x}{R}$ is defined as the rate of change of the desired orientation of the vehicle(rad/s).  $\Phi(x)$ is the neural network controller, and its output is
the front wheel steering angle (rad). System parameters are given in Table \ref{tab_example}.  By calculation to $u$, we set $\underline{u}=[-1,-3,-1,-1]^T,\overline{u}=[1,-1,1,0]^T$.

\begin{table}[]
    \centering
    \caption{System Parameters for Lateral Vehicle Control System.}
    \begin{tabular}{l|l}
   \hline\hline
       Total mass of vehicle  & $m=1573kg$ \\
       Yaw moment of inertia of vehicle  &  $I_z=2873kg\centerdot m^2$\\
       Longitudinal distance from c.g. to front tires & $l_f=1.1m$ \\
       Longitudinal distance from c.g. to rear tires & $l_r=1.58m$\\
       Front tire cornering stiffness & $C_{af}=80000N/rad$\\
       Rear tire cornering stiffness & $C_{ar}=80000N/rad$ \\
       Longitudinal velocity of the c.g. of the vehicle & $V_x=30m/s$
       \\
       Constant road radius & $R=400m$ \\
     \hline\hline
    \end{tabular}
    \label{tab_example}
\end{table}

The lateral vehicle control system we are discussing does not contain nonlinear functions, which means $f(x)=\underline{f}(\underline{x},\overline{x})=\overline{f}(\underline{x},\overline{x})=0$, so the corresponding interval observer system is as follows
\begin{equation}
\mathcal{M}_{eg}: \left\{
\begin{array}{ll}
\underline{\dot{x}}&=(A-\underline{L}C)\underline{x}+\underline{L}y+B_\Phi\underline{\Phi}(\underline{x},\overline{x})+\underline{u}\\
\dot{\overline{x}}&=(A-\overline{L}C)\overline{x}+\overline{L}y+B_\Phi\overline{\Phi}(\underline{x},\overline{x})+\overline{u}\\
\end{array}
\right. .
\notag
\end{equation}

The following describes how the auxiliary neural networks, $\underline{\Phi}(\underline{x},\overline{x})$\  \wesley{and}  $\overline{\Phi}(\underline{x},\overline{x})$, and the interval observer gains, $\underline{L} $ \wesley{and} $\overline{L}$, are obtained in this example. According to the feedback gain $K$ given in paper \citep{alleyne1997comparison}, the system (\ref{Formula 42}) can operate normally. Based on \wesley{this} operating data, we train the neural network controller $\Phi(x)$, which is parameterized by a 3-layer feedforward neural network with $n_1=5$, $n_2=5$, and $n_3=1$, and $\psi(v)=tanh(v)$ as the activation function of the first two layers\wesley{. T}he third layer does not use the activation function according to the settings in our paper.
The auxiliary neural networks $\underline{\Phi}(\underline{x},\overline{x})$ and $\overline{\Phi}(\underline{x},\overline{x})$ are designed based on $\Phi(x)$ according to (\ref{Formula 10}) and (\ref{Formula 11}). Considering the physical limitations of vehicle dynamics, the range of front wheel steering angles is limited to $[-\pi/6,\pi/6]$, which means that the output of neural network $\Phi(x)$ and auxiliary neural networks, $\underline{\Phi}(\underline{x},\overline{x})$ and $\overline{\Phi}(\underline{x},\overline{x})$, are limited to $[-\pi/6,\pi/6]$. The observer gains, $\underline{L}$ \wesley{and} $\overline{L}$, can be obtained by solving linear matrix inequalities (\ref{Formula 25}) and (\ref{Formula 26}) in Theorem \ref{Theorem 2}.

The run-time boundary estimations of state trajectories of lateral position error $\{e_1,\dot{e}_1\}$ and yaw angle error $\{e_2,\dot{e}_2\}$ during the lateral vehicle control system evolves in time interval $[0, 10]$ are shown in Figures \ref{FIG:3} and \ref{FIG:4}. The lateral position error and yaw angle error decrease significantly after the system reaches a steady state, indicating that the original system operates normally under the action of the neural network controller $\Phi(x)$. It is worth noting that the steady-state values of 
$e_1$ and $e_2$ are not zero \wesley{because} the input due to road curvature $\dot{\theta}_{des}$ is non-zero. The specific physical explanation of these 
steady-state errors can be found in Sections 3.2 and 3.3 of \citep{rajamani2011vehicle}. As shown in the results, the state trajectories (solid line) always run between the upper and lower bounds of the interval observer (dashed line), indicating that the interval observer we have designed can be used for state safety monitoring.

\begin{figure}[ht]
	\centering
		\includegraphics[scale=.43]{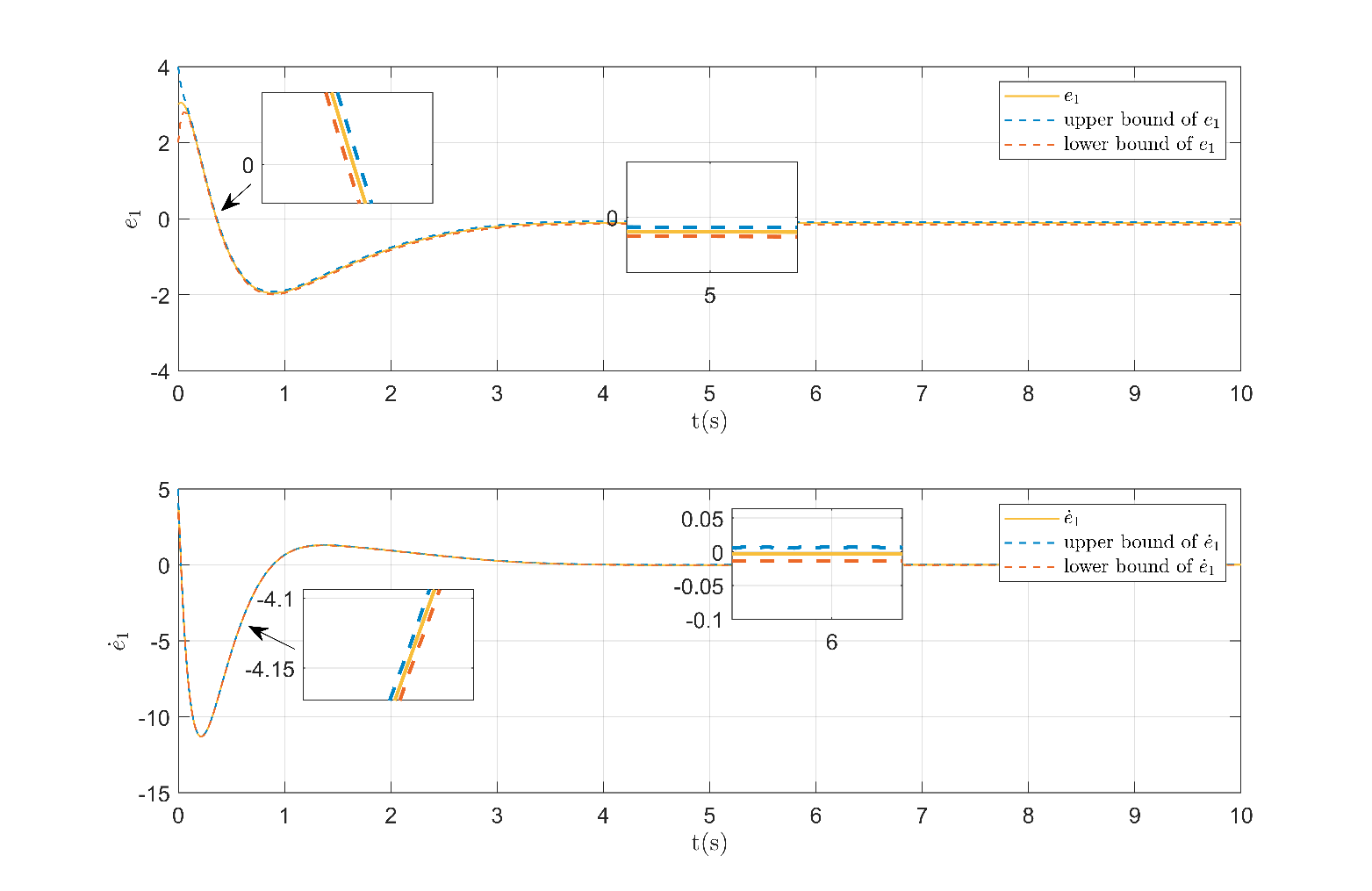}
	\caption{Safety monitoring of lateral position error $e_1$ and its derivative $\dot{e}_1$.}
	\label{FIG:3}
\end{figure}

\begin{figure}[ht]
	\centering
		\includegraphics[scale=.43]{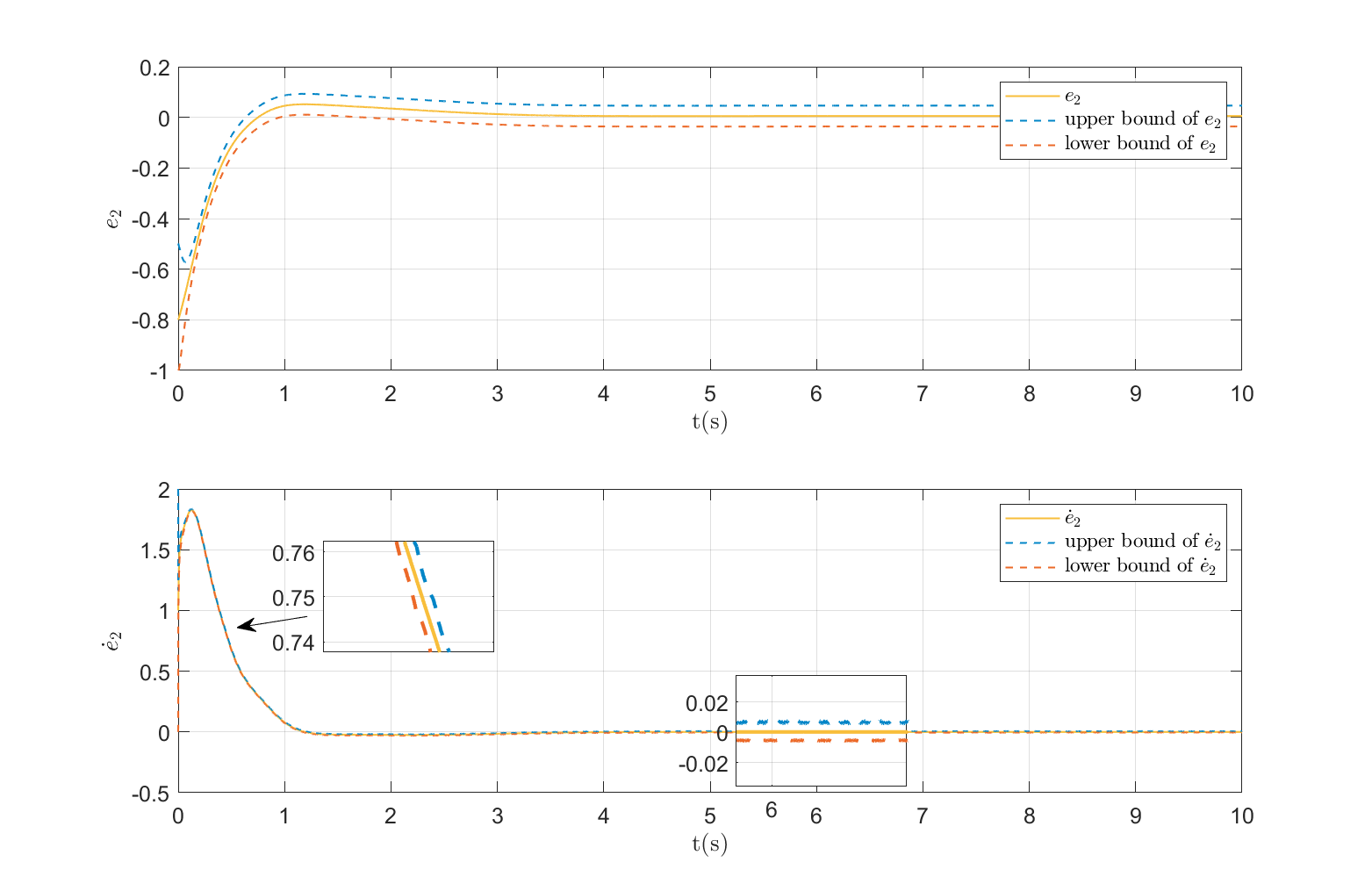}
	\caption{Safety monitoring of yaw angle error $e_2$ and its derivative $\dot{e}_2$.}
	\label{FIG:4}
\end{figure}

\section{Conclusions}
This paper presents a possible solution to the problem of run-time safety monitoring of dynamical systems embedded with neural network components. A design approach for a safety monitor is proposed for the system characteristics. The safety monitor works as a Luenberger-type interval observer, which estimates the upper and lower bounds of the state run-time trajectory in real time. The design process of the interval observer consists of two main components: the two auxiliary neural networks and the observer gain. The two auxiliary neural networks can be obtained from the neural network embedded in the original system. The presence of nonlinear activation functions in neural networks makes it difficult to apply traditional control theory to calculate observer gains $\underline{L}$ \wesley{and} $\overline{L}$.
To solve this problem, we use quadratic constraints (QCs) to abstract the nonlinear activation functions in neural networks. The computational problem of observer gain is expressed in a series of convex optimization problems. The interval observer design method is applied to the lateral vehicle control system to verify the correctness of the proposed solutions. The correction of neural network operation in the event of security problems needs to be considered in future work. Further applications to dynamical systems with more complex behaviors such as switched or hybrid systems  \citep{zhu2019quasi,xiang2017output,li2020neural} will be also considered in the future.

\bibliographystyle{apacite}
\bibliography{ref}

\end{document}